\documentclass[10pt, journal, twoside, web, letterpaper]{IEEEtran}

\usepackage{generic}
\usepackage{cite}
\usepackage[compact]{titlesec}
\usepackage{amsmath,amssymb,amsfonts}
\usepackage{hyperref}
\usepackage{graphicx}
\usepackage{textcomp}
\usepackage{times}
\usepackage{mathrsfs}
\usepackage{amstext} 
\usepackage{algorithm}
\usepackage{algorithmic}
\usepackage{setspace} 
\usepackage{amssymb} 
\usepackage{mathtools} 
\usepackage{verbatim}
\usepackage{booktabs} 
\usepackage[font=footnotesize]{caption}
\usepackage[font=footnotesize]{subcaption}
\renewcommand{\theenumi}{\alph{enumi}}

\newcommand{\real}{\mathbb{R}}

\newcommand{\Ec}{\mathcal{E}}
\newcommand{\Fc}{\mathcal{F}}

\newcommand{\Hc}{\mathcal{H}}

\newcommand{\Kc}{\mathcal{K}}

\newcommand{\Pc}{\mathcal{P}}

\newcommand{\Rc}{\mathcal{R}}
\newcommand{\Sc}{\mathcal{S}}

\newcommand{\Uc}{\mathcal{U}}
\newcommand{\Vc}{\mathcal{V}}
\newcommand{\Wc}{\mathcal{W}}

\newcommand{\rank}[1]{\operatorname{rank}(#1)}

\newcommand{\innerprodF}[2]{\langle #1, #2 \rangle_{\Fc}}

\newcommand{\longthmtitle}[1]{\mbox{}{\textit{(#1):}}}

\newcommand{\oprocendsymbol}{\hbox{$\square$}}
\newcommand{\oprocend}{\relax\ifmmode\else\unskip\hfill\fi\oprocendsymbol}

\newtheorem{theorem}{Theorem}[section]
\newtheorem{definition}[theorem]{Definition}

\newtheorem{lemma}[theorem]{Lemma}

\newtheorem{remark}[theorem]{Remark}

\newtheorem{corollary}[theorem]{Corollary}
\newtheorem{proposition}[theorem]{Proposition}  
\newtheorem{problem}[theorem]{Problem}

\renewcommand{\baselinestretch}{.995}
\graphicspath{{fig/}}

\usepackage{multirow}  

\makeatletter
\newcommand{\RightComment}[1]{%
  \unskip\hfill{\footnotesize\(\triangleright\)~#1}%
}
\makeatother

\allowdisplaybreaks

\def\BibTeX{{\rm B\kern-.05em{\sc i\kern-.025em b}\kern-.08em
    T\kern-.1667em\lower.7ex\hbox{E}\kern-.125emX}}

\begin{document}


\title{\LARGE Subspace Pruning via Principal Vectors
  \\
  for Accurate Koopman-Based Approximations%
  \thanks{A preliminary version of this work has been submitted
    as~\cite{DS-JC:26-cdc1} to the 2026 IEEE Conference on Decision and
    Control.
  }}

\author{Dhruv Shah \quad Jorge Cort\'es \thanks{D. Shah and J. Cort\'es
    are with Department of Mechanical and Aerospace Engineering, UC
    San Diego, USA, {\tt\small \{dhshah,cortes\}@ucsd.edu}}}

\maketitle

\thispagestyle{empty}
\begin{abstract}
  The accuracy of Koopman operator approximations over
  finite-dimensional spaces relies critically on their invariance
  properties. These can be rigorously quantified via the principal
  angles between a candidate subspace and its image under the Koopman
  operator.  This paper proposes a unified algebraic framework for
  subspace pruning designed to systematically refine the invariance
  error. We establish the geometric equivalence between
  consistency-based methods and principal-vector pruning, and build on
  this insight to introduce a hybrid strategy that balances between
  multiple and single principal vector pruning for improved numerical
  stability and scalability. We derive error bounds for the retention
  of approximate and external eigenfunctions, demonstrating that the
  multi-vector approach mitigates the numerical drift inherent to
  sequential pruning. To ensure scalability, we develop an efficient
  numerical update scheme based on rank-one modifications that reduces
  the computational complexity of tracking principal angles by an
  order of magnitude.  Finally, we exploit the subspace obtained from
  the pruning algorithms to build a lifted linear model for state
  prediction that accounts for the trade-offs between improving
  invariance and minimizing state reconstruction error. Simulations
  demonstrate the effectiveness of our approach.
\end{abstract}


\section{Introduction}


%
%

The Koopman operator perspective~\cite{BOK-JVN:32} on dynamical
systems owes its popularity to mapping nonlinear state-space evolution
to a linear operator acting on an infinite-dimensional vector space of
observable functions. This enables the application of powerful
spectral analysis techniques to understand complex nonlinear systems
and provides a unified linear structure that can be utilized for
control. However, the practical implementation of these techniques is
inherently complicated by its infinite-dimensional
nature. Specifically, accurately approximating Koopman eigenfunctions
(a critical requirement for reliable long-term forecasting) remains a
significant challenge.  Techniques that rely exclusively on minimizing
one-step residual errors frequently fail to capture the underlying
spectral geometry. Consequently, these methods often yield poor
long-term predictions, even when their single-step performance appears
satisfactory.  This paper develops an algebraic framework to construct
finite-dimensional Koopman models for long-term predictive
fidelity. By guaranteeing the structural invariance of these models,
our approach directly addresses the critical need for reliable
long-horizon forecasting and accurate eigenfunction approximation
-- both essential for data-driven control and closed-loop stability
analysis in complex nonlinear systems.


\subsection*{Literature Review}
The Koopman operator has been a subject of intense research interest,
with applications across diverse domains, including model reduction
\cite{SK-FN-SP-JHN-CC-CS:20}, fluid dynamics \cite{AD-DL-DS-AT:21},
power grids \cite{SPN-SG-SK-SP-KA-YW-SC:22}, and robotics
\cite{LS-MH-GM-DB-IA-TM-JC-KK:26}. In systems and control, theoretical
contributions include global stability criteria for attractors based
on eigenfunctions \cite{AM-IM:16}, connections to contraction theory
\cite{BY-IRM:24}, data-driven approximation of Lyapunov functions
\cite{SAD-DVD:23}, stability of switched systems \cite{CMZ-AM:25}, and
the estimation of regions of attraction~\cite{YM-RZ-JL:23}.  While the
original theory addresses autonomous systems, extensions to control
systems \cite{MK-IM:20,MH-JC:26-auto, MH-IM-JC:25-tac} have enabled
the integration of classical control tools, including feedback
linearization \cite{DG-VK-FP:24}, control Lyapunov functions
\cite{VZ-EB:23}, and optimal control
\cite{MEV-CNJ-BH:21,JH-KC:24}. Significant advancements in the control
of Koopman models have been achieved by leveraging robust control
techniques to address design challenges. Notably, the SafEDMD
framework and Sum-of-Squares (SoS) optimization provide rigorous
closed-loop stability guarantees even under model
mismatch~\cite{RS-MS-KW-JB-FA:26,RS-JB-FA:25}. Complementing these
stability results, recent work~\cite{DS-JC:25-csl} targets controller
synthesis specifically for input-state separable Koopman models of
non-affine control systems. For a recent overview, we refer the reader
to the comprehensive survey~\cite{RS-KW-IM-JB-MS-FA:26}.

%
%

Given the operator's infinite-dimensional nature, its practical
implementation requires finite-dimensional representations. However,
because exact linear representations are rarely available for complex
systems~\cite{ZL-NO-EDS:25}, approximations are typically constructed
by projecting the operator onto a chosen subspace. Prominent examples
of this approach include Dynamic Mode Decomposition
(DMD)~\cite{PJS:10} and Extended DMD (EDMD) \cite{MOW-IGK-CWR:15},
which utilize data to project dynamics onto a subspace spanned by a
predefined dictionary. The theoretical properties of EDMD are
well-established, including convergence in asymptotic
regimes~\cite{MK-IM:18} and probabilistic error bounds for finite
data~\cite{FN-SP-FP-MS-KW:23}. More recently, this theoretical
framework has expanded to include $L^{\infty}$ error bounds for kernel
EDMD~\cite{FK-FMP-MS-AS-KW:25} and extensions to stochastic
systems~\cite{MH-FMP-MS-KW:25}.

Approximations via projection introduce two distinct but related
issues: the emergence of spurious eigenfunctions and truncation errors
in the prediction of arbitrary functions. While spurious
eigenfunctions are a critical concern in spectral analysis
applications \cite{SK-FN-SP-JHN-CC-CS:20,AD-DL-DS-AT:21},
prediction-focused applications require bounds on the approximation
error for the entire subspace \cite{MH-JC:23-auto}. These errors
vanish if the subspace is invariant under the Koopman operator,
motivating the search for (near-)invariant subspaces.  Approaches to
this problem of identifying (near-)invariant subspaces
%
%
include constructing subspaces from approximated eigenfunctions
\cite{MK-IM:20} or generalized eigenfunctions \cite{MJC-CD-AH:25}, and
learning subspaces via optimization using neural networks
\cite{YX-KS-NKL-ZS:25,BG-JP-SC-OA:25}.  Other relevant research
focuses on spectral pollution and spurious eigenvalues using Ritz
residuals \cite{HZ-STMD-CWR-EAD-LNC:20} or sparsity-promoting
optimization \cite{SI-SB-AA-NK:25}.  Residual Dynamic Mode
Decomposition (ResDMD) \cite{MJC-LJA-MS:23,MJC-AT:24,MJK:24} also
employs residuals to manage spectral pollution and estimate
pseudospectra.  Building on the computational tools introduced
there,~\cite{GC-NB-JCL-SB-MJC:26} has recently introduced,
independently from the present manuscript, the data-driven Principal
Angle Decomposition (PAD) algorithm, which proceeds by identifying
principal angles and vectors between a subspace and its image under
the Koopman operator, and orders and retains small-angle principal
observables to construct reduced, more accurate Koopman models.

An alternative line of research follows two natural, intersecting
tracks: accuracy measures for Koopman approximations and
accuracy-improving pruning algorithms. For the former, the concept of
invariance proximity evolved from $\epsilon$-apart subspaces
\cite{MH-JC:23-auto}, followed by the consistency index
\cite{MH-JC:23-csl},
to general inner product spaces with control applications
\cite{MH-JC:26-auto}. For general spaces of functions, invariance
proximity is the sine of the maximum principal angle between a
subspace and its image under the Koopman operator,
cf. \cite{MH-JC:24-csl-arxiv-revised}. For the latter track,
algorithms seek to iteratively prune candidate dictionaries to yield
subspaces with improved invariance properties. This includes Symmetric
Subspace Decomposition (SSD) \cite{MH-JC:22-tac} for finding the
largest invariant subspace, its tunable variant (T-SSD)
\cite{MH-JC:23-auto}, and Recursive Forward-Backward EDMD (RFB-EDMD)
\cite{MH-JC:25-access}.
  
\subsection*{Statement of Contributions}
The key contributions of the paper are as follows:

\noindent
\textbf{(1) geometric unification of pruning algorithms:} we introduce
the Single Principal Vector (SPV) pruning algorithm to iteratively
remove principal vectors and minimize invariance proximity of a
subspace of functions.  We establish the algebraic equivalence between
the consistency-based pruning of RFB-EDMD and SPV pruning,
demonstrating that the consistency matrix eigenvectors are identical
to the principal vectors of the current subspace;

\textbf{(2) two new pruning algorithms:} we introduce the Multiple
Principal Vector (MPV) algorithm for simultaneous multiple-vector
pruning and a hybrid MPV-SPV strategy that effectively reduces
cumulative numerical drift. This hybrid approach is particularly
beneficial for large dictionaries of functions, where sequential SPV
pruning can lead to significant numerical drift and lead to loss of
eigenfunctions retained;

\textbf{(3) external eigenfunction retention:} our analysis
establishes that true Koopman eigenfunctions in the initial dictionary
are strictly preserved, and we derive retention bounds for
``external'' eigenfunctions within an $\epsilon$-neighborhood. In
practice, when working with finite precision, we only have approximate
eigenfunctions, and the multi-vector approach significantly improves
the retention of these approximate eigenfunctions by mitigating
numerical drift;

\textbf{(4) efficient computation:} we implement a fast numerical
scheme using rank-one updates and incremental QR decompositions,
reducing the complexity of recomputing principal angles by an order of
magnitude. Numerical experiments demonstrate that the runtime of the
MPV-SPV algorithm is dominated by the first SVD computation,
indicating near minimal overhead for subsequent pruning steps;

\textbf{(5) state reconstruction:} we describe a framework for
constructing lifted linear models for state prediction of unknown
dynamical systems. This allows the designer to balance invariance
proximity with the ability to reconstruct the original state space,
which is critical for control applications.
%
%

A preliminary, conference version of this article has been
submitted as \cite{DS-JC:26-cdc1}.
The present work significantly extends~\cite{DS-JC:26-cdc1} by
developing two new pruning algorithms with improved numerical stability and
computational efficiency via rank-one updates. Our external eigenfunction retention analysis is also novel, it provides an explanation for the curious numerical drift observed in the SPV pruning algorithm, and it establishes that the multi-vector approach significantly mitigates this drift, leading to improved retention of approximate eigenfunctions. 
%
%
Additionally, we provide a novel method for balancing invariance with
state reconstruction.

\subsection*{Notation}
Let $\real^n$ denote the $n$-dimensional Euclidean space and
$\mathcal{X} \subseteq \real^n$ the state space. For a matrix
$M \in \real^{m \times n}$, we denote its range (column space) by
$\mathcal{R}(M)$, its rank by $\rank{M}$, and its Moore-Penrose
pseudoinverse by $M^\dagger$. For subspaces
$\Uc, \Vc \subseteq \real^n$, we denote the orthogonal projection onto
$\Uc$ by $\Pc_{\Uc}$, and the orthogonal complement by $\Uc^\perp$. We
write $\Uc \oplus \Vc$ to denote the orthogonal direct sum when
$\Uc \cap \Vc = \{0\}$ and $\Uc \perp \Vc$. For a Hilbert space $\Hc$
with inner product $\langle \cdot, \cdot \rangle_{\Hc}$ and induced
norm $\|\cdot\|_{\Hc}$, subspaces of observables are denoted by
calligraphic letters (e.g., $\Sc, \Vc$). We use $\text{span}(\cdot)$
to denote the linear span and $\text{dim}(\cdot)$ to denote
dimension. The distance from a point $x$ to a subspace $\Uc$ is
$\text{dist}(x, \Uc) = \inf_{u \in \Uc} \|x - u\|$.

\section{Preliminaries}

In this section, we review the theoretical foundations of the Koopman
operator and its finite-dimensional approximations. We specifically
formulate Extended Dynamic Mode Decomposition (EDMD) as an orthogonal
projection. We then take a brief detour to review the geometric
concepts of principal angles and vectors. Leveraging these geometric
tools, we introduce the concept of invariance proximity to quantify
the quality of Koopman approximations.

\subsection{The Koopman Operator}
Following~\cite{AM-YS-IM:20}, consider a discrete-time dynamical
system on the state space $\mathcal{X} \subseteq \mathbb{R}^n$
described by a map $T: \mathcal{X} \to \mathcal{X}$:
\begin{equation}\label{eqn:sys_dynamics}
  x^+ = T(x), \quad x \in \mathcal{X}.
\end{equation}
The Koopman operator $\mathcal{K}: \mathcal{F} \to \mathcal{F}$ is an
infinite-dimensional linear operator that acts on a space of
real-valued observables
$\mathcal{F} \ni \psi: \mathcal{X} \to \mathbb{R}$ by composing them
with the dynamics:
\begin{equation}
    (\mathcal{K}\psi)(x) = \psi(T(x)).
\end{equation}
We assume that the function space $\mathcal{F}$ is closed under
composition with $T$. Although the state-space map $T$ may be
nonlinear, the Koopman operator $\mathcal{K}$ acts linearly on
observables. This linearity enables the use of spectral methods: the
eigenvalues and eigenfunctions of $\mathcal{K}$ encode the long-term
asymptotic behavior of the system through the evolution of observable
functions. 
%
%

\subsection{EDMD as an Orthogonal Projection}\label{sec:EDMD}
Since $\mathcal{K}$ operates on an infinite-dimensional space
$\mathcal{F}$, practical implementations must approximate it on
finite-dimensional subspaces. We equip the space of observables
$\mathcal{F}$ with a Hilbert space structure by defining an inner
product $\langle \cdot, \cdot \rangle_{\mathcal{F}}$ and the
associated norm $\|\cdot\|_{\mathcal{F}}$. A common choice is the
$L_2$-space with respect to a probability measure $\mu$ on
$\mathcal{X}$. Let $\mathcal{S} \subset \mathcal{F}$ be a subspace
spanned by a finite set of linearly independent functions
$\Psi = \{\psi_1, \dots, \psi_s\}$, which we refer to as a dictionary.


The goal of data-driven approximation is to find a finite-dimensional
operator $K: \mathcal{S} \to \mathcal{S}$ that best represents the
action of $\mathcal{K}$ when restricted to $\mathcal{S}$. Extended
Dynamic Mode Decomposition (EDMD) \cite{lQL-FD-EMB-IGK:17} provides
the optimal approximation in the $L_2$-sense by orthogonally
projecting the image of the subspace back onto itself.  Formally, let
$P_{\mathcal{S}}: \mathcal{F} \to \mathcal{S}$ denote the orthogonal
projection operator onto $\mathcal{S}$. The EDMD approximation is
given by $K_{\text{EDMD}} \triangleq P_{\mathcal{S}} \mathcal{K}|_{\mathcal{S}}$.

Given a dataset of snapshot pairs $\{(x_i, x_i^+)\}_{i=1}^N$ where
$x_i^+ = T(x_i)$, we construct the data matrices
$\Psi(X), \Psi(X^+) \in \mathbb{R}^{M \times s}$, where the $i$-th rows
are the evaluations of the dictionary functions at $x_i$ and $x_i^+$
respectively.
%
%
The matrix representation of $K_{\text{EDMD}}$ in the
basis $\Psi$ is the solution to the least-squares problem
\begin{align}
  \label{eq:LS-EDMD}
  \min_K \|\Psi(X) K - \Psi(X^+)\|_F,   
\end{align}
given explicitly by
$K = \Psi(X)^\dagger \Psi(X^+) \in \mathbb{R}^{s \times s}$.

This projection interpretation of EDMD reveals a critical limitation:
if the subspace $\mathcal{S}$ is not invariant under $\mathcal{K}$
(i.e., $\mathcal{K}\mathcal{S} \not\subseteq \mathcal{S}$), the
projection $P_{\mathcal{S}}$ discards the component of the dynamics
that evolves orthogonal to $\mathcal{S}$, leading to approximation
errors.

\subsection{Principal Angles and Vectors}\label{sec:principal-angles}
Here we recall the basic definitions and properties of principal
angles and vectors~\cite{AB-GHG:73},
%
%
which will be useful to rigorously quantify the alignment between the
chosen subspace and its evolution under the Koopman operator.

\begin{definition}[Principal Angles and Vectors]
\label{defn:pa_pv}
Let \((\mathcal{H},\langle\cdot,\cdot\rangle)\) be a Hilbert space,
and let \(\mathcal{U},\mathcal{V}\subset\mathcal{H}\) be subspaces
with \(\dim(\mathcal{U})=d_1\) and \(\dim(\mathcal{V})=d_2\).  The
principal angles
$0 \leq \theta_1 \leq \cdots \leq \theta_k \leq \frac{\pi}{2}$
between $\mathcal{U}$ and $\mathcal{V}$, where
$k = \min\{d_1, d_2\}$, are defined recursively as follows:
\begin{align*}
  \cos \theta_j = \max_{u \in \mathcal{U},\, v \in \mathcal{V}} \;
  &
    |\langle
    u,
    v
    \rangle|
  \\ 
  \text{subject to} \quad & \|u\| = \|v\| = 1,
  \\
  & \langle u, u_i \rangle = 0, \; \langle v, v_i \rangle = 0, \,\,
    i = 1, \ldots, j-1, 
\end{align*}
where $u_i, v_i$ are the principal vectors corresponding to the
previous $(j-1)$ angles. The vectors $(u_j, v_j)$ achieving the
maximum are called the $j$-th pair of principal vectors.
\end{definition}

\begin{remark}
  The principal vectors $\{u_j\}_{j=1}^k$ and $\{v_j\}_{j=1}^k$ are
  orthonormal, i.e.,
  \begin{equation*}  
    \langle u_i,u_j\rangle=\delta_{ij},\quad
    \langle v_i,v_j\rangle=\delta_{ij},\quad i,j=1,\dots,k.
  \end{equation*}
  In particular, the principal vectors can be extended to bases of
  their subspaces. For instance, if $k=\dim(\Uc)$, then
  $\{u_j\}_{j=1}^k$ is an orthonormal basis of $\Uc$, and
  $\{v_j\}_{j=1}^k$ can be augmented to an orthonormal basis of
  $\Vc$. \oprocend
\end{remark}

One can show, cf.~\cite[Proposition~4.4]{MH-JC:24-csl-arxiv-revised}, that the
principal vectors $\{u_j\}_{j=1}^k$ and $\{v_j\}_{j=1}^k$ satisfy
$\langle u_i,v_j\rangle=\delta_{ij} \, \cos \theta_i$, for all
$i,j=1,\dots,k$. Next, we describe how to compute the principal angles
and vectors in the Euclidean setting, i.e., $\mathcal{H} = \real^n$,
via the Singular Value Decomposition (SVD).
%
%

\begin{theorem}[Computation via SVD
  \cite{AB-GHG:73}]\label{thm:Golub1973}
  Let $\Uc,\Vc\subset\real^n$ have orthonormal basis matrices
  $Q_{\Uc}\in\real^{n\times d_1}$ and $Q_{\Vc}\in\real^{n\times d_2}$,
  and set $k=\min\{d_1,d_2\}$.  Compute the compact SVD
  %
  %
  \[
    \widetilde U \, \Sigma \, \widetilde V^\top \;=\; Q_{\Uc}^\top Q_{\Vc} ,
    \qquad
    \Sigma=\operatorname{diag}(\sigma_1,\ldots,\sigma_k),
  \]
  %
  %
  where $\sigma_1\ge\cdots\ge\sigma_k\ge0$.  Then, the principal
  angles $\{\theta_j\}_{j=1}^k$ and vectors $\{u_j\}_{j=1}^k$,
  $\{v_j\}_{j=1}^k$ between $\Uc$ and $\Vc$ satisfy
  \begin{equation}
    \cos\theta_j=\sigma_j,\,u_j = Q_{\Uc}\,\widetilde u_j,\, v_j =
    Q_{\Vc}\,\widetilde v_j,  j=1,\dots,k,   
  \end{equation}
  where $\widetilde u_j$ and $\widetilde v_j$ denote the $j$-th
  columns of $\widetilde U \in \real^{d_1\times k}$ and
  $\widetilde V \in \real^{d_2\times k}$. If $\sigma_j=0$, then
  $\theta_j=\tfrac{\pi}{2}$, and the corresponding vectors may be
  chosen from the appropriate nullspaces.
\end{theorem}

\subsection{Invariance Proximity}
As reasoned in Section~\ref{sec:EDMD}, to ensure accuracy of the
Koopman approximation via EDMD, the chosen finite-dimensional subspace
$\mathcal{S}$ should be as close to invariant as possible. We quantify
this ``closeness'' using the geometric concept of principal angles.

Let
$\mathcal{K}\mathcal{S} = \text{span}\{\mathcal{K}\phi \mid \phi \in
\mathcal{S}\}$ denote the image of the subspace under the
operator. The proximity of $\mathcal{S}$ to invariance is determined
by the alignment between $\mathcal{S}$ and $\mathcal{K}\mathcal{S}$.
Let $0 \le \theta_1 \le \dots \le \theta_s \le \pi/2$ be the principal
angles between $\mathcal{S}$ and $\mathcal{K}\mathcal{S}$. These
angles recursively maximize the inner product between unit vectors in
the two subspaces.

\begin{definition}[Invariance
  Proximity~\cite{MH-JC:24-csl-arxiv-revised}]\label{definition:IPT} 
  The invariance proximity of a subspace $\mathcal{S}$ with respect to
  the operator $\mathcal{K}$ is defined by
  \begin{equation}
    \delta(\mathcal{S}) \triangleq \sin \theta_{\max}(\mathcal{S},
    \mathcal{K}\mathcal{S}). 
  \end{equation}
\end{definition}
%
%

A value of $\delta(\mathcal{S}) = 0$ means that
$\mathcal{K}\mathcal{S} \subseteq \mathcal{S}$, indicating that
$\mathcal{S}$ is an invariant subspace.  The following result
clarifies the practical significance of invariance proximity by
showing that it exactly corresponds to the worst-case prediction error
of the EDMD model over the subspace.

\begin{theorem}[Worst-Case Relative Prediction Error
  \cite{MH-JC:24-csl-arxiv-revised}]
  Let $\mathcal{S} \subset \mathcal{F}$ be a finite-dimensional
  subspace and let
  $K_{\text{EDMD}} = P_{\mathcal{S}} \mathcal{K}|_{\mathcal{S}}$ be
  the EDMD approximation of the Koopman operator on
  $\mathcal{S}$. Then,
  \begin{equation}
    \delta(\mathcal{S}) = \sup_{\substack{f \in \mathcal{S}
        \\
        \|\mathcal{K}f\| \neq 0}} \frac{\| \mathcal{K}f -
      K_{\text{EDMD}}f \|}{\| \mathcal{K}f \|}. 
  \end{equation}
\end{theorem}

This result implies that minimizing invariance proximity
$\delta(\mathcal{S})$ directly minimizes the maximum relative error
incurred by the EDMD predictor for any function in the subspace.

\section{Problem Statement}


Our aim is to obtain useful Koopman-based finite-dimensional
approximations of unknown dynamical
systems. Following\cite{MH-JC:25-access}, we formulate this objective
as a subspace search problem within the space of observables using the
notion of invariance proximity. Consider the underlying dynamical system \eqref{eqn:sys_dynamics} with the associated Koopman operator $\mathcal{K}$, let
$\mathcal{S}_{\text{init}}$ be an initial finite-dimensional subspace
spanned by a (large) dictionary of candidate functions. Ideally, we
seek a subspace $\mathcal{S} \subseteq \mathcal{S}_{\text{init}}$ that
is invariant under $\mathcal{K}$ and helps us
reconstruct the state effectively. However, exact invariance is rare
in finite dictionaries. Instead, we seek a subspace where the dynamics
are confined within the subspace to a user-specified degree of
accuracy. Furthermore, we want this subspace to be of the largest
possible dimension to maximize the expressivity of the Koopman
model. This can be formulated using invariance proximity as follows.

\begin{problem}[Invariant Subspace Search]\label{problem:subspace_search}
  Given an initial subspace $\mathcal{S}_{\text{init}}$ and a
  tolerance $\epsilon \in [0, 1)$, find a subspace
  $\mathcal{S}^* \subseteq \mathcal{S}_{\text{init}}$ of the largest
  possible dimension such that:
  \begin{equation}\label{eq:problem-bound}
    \delta(\mathcal{S}^*) \le \epsilon.
  \end{equation}
\end{problem}
\smallskip

Solving Problem \ref{problem:subspace_search} exactly requires a
combinatorial search over all possible subspaces of the initial
dictionary, which is computationally prohibitive. Consequently, we
focus on developing efficient algorithms to identify high-dimensional
subspaces that satisfy the invariance proximity constraint, even if
they are not the one with the largest dimension. Our primary
motivation behind optimizing the dimension of the retained subspace is
to preserve expressivity. If the resulting subspace $\mathcal{S}^*$ is
insufficient for accurate state reconstruction, it necessitates either
enriching the initial dictionary $\mathcal{S}_{\text{init}}$ or
relaxing the tolerance~$\epsilon$.

%
%

\begin{remark}[Beyond One-Step Prediction]
  Here, we justify why minimizing invariance proximity is preferable
  over standard residual error minimization to identify Koopman
  models.  Standard data-driven methods, such as EDMD, approximate the
  operator $\Kc$
    %
    %
    by minimizing the one-step prediction error on a fixed
    dictionary, as in~\eqref{eq:LS-EDMD}.
    %
    %
    While this approach ensures that the model fits the data over a single
    time step, it does not guarantee that the subspace spanned by $\Psi$
    is invariant.
    Instead, we argue that  minimizing the
    one-step RMSE 
    is insufficient for the following key considerations:
    \begin{LaTeXdescription}
      
    \item[Long-Term Prediction:] Low one-step error does not imply
      accurate long-term prediction. If the subspace is not invariant, the
      trajectory of the observables  ``leak'' out of the subspace upon
      iteration. This leakage accumulates geometrically, rendering
      multi-step predictions unreliable. In contrast, minimizing 
      invariance proximity  explicitly minimizes this
      leakage, ensuring that the linear model remains valid over long
      horizons;
      
    \item[The Value of Linearity:] If the primary goal were simply to
      minimize one-step prediction error, restricting the model to a linear
      map would be counterproductive; non-linear approximators (e.g., neural
      networks or Gaussian processes) usually offer superior one-step
      expressivity.
      %
      %
    \end{LaTeXdescription}
    The formulation in Problem~\ref{problem:subspace_search} with
    invariance proximity thereby seeks to give rise to good long term
    prediction while retaining the linearity structure of the model.
    \oprocend
\end{remark}

The algorithms presented below are designed to solve
Problem~\ref{problem:subspace_search} by iteratively pruning
directions from $\mathcal{S}_{\text{init}}$ that violate the
invariance condition, thereby systematically improving the value of
the invariance proximity.
Another way of minimizing invariance proximity would be to learn the
dictionary functions $\Psi$ (e.g., via neural networks) to directly
reduce $\delta$. However, such objectives are typically non-convex,
computationally costly, and lack deterministic guarantees. Instead,
the subspace-pruning strategies adopted here exploit the linear
structure of the function space to algebraically identify and remove
leaky directions in a fixed dictionary, yielding efficient algorithms
with deterministic performance guarantees.

%
%

The remainder of this paper is organized as follows. In
Section~\ref{sec:pruning_algorithms}, we introduce the Principal
Vector based pruning framework, detailing the baseline Single-PV (SPV)
algorithm. We prove its algebraic equivalence to the Recursive
Forward-Backward EDMD (RFB-EDMD) method~\cite{MH-JC:25-access} and
introduce the Multiple-PV (MPV) algorithm, which accelerates pruning
by dropping multiple directions simultaneously. In
Section~\ref{sec:properties_pruning}, we analyze the stability
properties of these algorithms and derive error bounds for the
retention of ``external'' eigenfunctions to compare the numerical
performance of MPV and SPV strategies. Motivated by this analysis, we
introduce a hybrid MPV-SPV pruning algorithm that combines the
advantages of both approaches. We then address computational
efficiency in Section~\ref{sec:efficient_computation}, presenting a
numerical scheme based on rank-one updates and incremental QR
decompositions that reduces the cost of iterative pruning by an order
of magnitude.
Finally, in
Section~\ref{sec:state_construction}, we detail the construction of
the reduced-order linear model and propose a decoupled architecture
using an auxiliary lifting map to recover the state without
compromising the invariance of the dynamical model.

\section{Subspace Pruning via Principal
  Vectors} \label{sec:pruning_algorithms}

In this section, we propose a systematic approach to solve Problem
\ref{problem:subspace_search} by iteratively refining the
dictionary. The core idea relies on the geometric interpretation of
invariance proximity. Recall that $\delta(\mathcal{S})$ is determined
by the largest principal angle between the subspace $\mathcal{S}$ and
its image $\mathcal{K}\mathcal{S}$,
cf. Definition~\ref{definition:IPT}. Consequently, the principal
vectors associated with these large angles identify the specific
directions within $\mathcal{S}$ that are most responsible for
violating invariance—effectively, the directions where the dynamics
``leak'' out of the subspace.

Motivated by this insight, our strategy is to identify and remove
(prune) these high-error directions from the dictionary. We begin by
introducing a baseline algorithm that removes the single worst
direction at each step, termed Single-Principal-Vector (SPV)
pruning. We then establish its theoretical connection to existing
methods before generalizing it to a multi-vector pruning
strategy. First, we clarify a standing assumption and notation that
will be used throughout the paper.

\vspace{0.1in}
\noindent \textbf{Standing Assumption:} Throughout the paper, we
consider a finite-dimensional subspace
$\mathcal{S} \subset \mathcal{F}$ and its image
$\mathcal{K}\mathcal{S}$ under the Koopman operator, with
$\dim(\mathcal{S}) = \dim(\mathcal{K}\mathcal{S}) = s$. We denote
their principal angles, arranged in increasing order, and the
corresponding principal vectors by
\begin{equation}
\{\theta_i \}_{i=1}^s \subset [0, \pi/2], \, \{u^{\Sc}_i\}_{i=1}^s \subset \Sc, \, \{\Kc v_i^{\Kc \Sc}\}_{i=1}^s \subset \Kc \Sc.
\label{eqn:setting_principal_arguments}
\end{equation}
These arguments are computed with respect to the inner product
$\langle \cdot, \cdot \rangle_{\mathcal{F}}$ on $\mathcal{F}$, for
example the $L_2$ inner product induced by the data measure. We omit
the subscript when the context is clear.  \vspace{0.1in}

\subsection{Computation of Principal Arguments in
  $L_2(\mu_X)$}
We describe here how to compute principal angles and vectors in the
data-driven setting, which is critical for the pruning
procedures. This involves defining an appropriate inner product on the
space of observables and leveraging the isomorphism between function
spaces and finite-dimensional Euclidean spaces induced by the data.

Consider the nonlinear system~\eqref{eqn:sys_dynamics}. We utilize a
dataset of $N$ snapshot pairs organized into data matrices
$X, X^+ \in \mathbb{R}^{N \times n}$, where
$X = [x_1, \dots, x_N]^\top$ and $X^+ = [x_1^+, \dots, x_N^+]^\top$
with $x_i^+ = T(x_i)$ for $i=1,\dots,N$. We fix an initial dictionary
of observables $\Psi = [\psi_1, \dots, \psi_s]$ and define the lifted
data matrices $A = \Psi(X) \in \mathbb{R}^{N \times s}$ and
$B = \Psi(X^+) \in \mathbb{R}^{N \times s}$. In this data-driven
framework, we equip the space of observables with the $L_2(\mu_X)$
inner product induced by the empirical data measure $\mu_X = \frac{1}{N} \sum_{i=1}^N \delta_{x_i}$. Specifically, for observables $f,g \in L_2(\mu_X)$, this inner product
is given by
$$
\langle f, g \rangle_{L_2(\mu_X)} = \int_{\mathcal{X}} f(x) g(x) \,
d\mu_X(x) = \frac{1}{N} \sum_{i=1}^N f(x_i) g(x_i).
$$
Under this inner product, we define the discrete evaluation map
$\mathcal{E}: L_2(\mu_X) \rightarrow \mathbb{R}^N$ as:
\begin{equation}
  \mathcal{E}(f) = \begin{bmatrix} f(x_1) \, \cdots \, f(x_N) \end{bmatrix}^{\top}.
\label{eq:isomorpshism_map}  
\end{equation}
The evaluation map $\Ec|_{\Sc}$ forms an isomorphism between the
candidate function subspace $\mathcal{S}$ and the Euclidean column
space $\mathcal{R}(A)$, and similarly the map $\Ec|_{\Kc \Sc}$ forms
an isomorphism between $\mathcal{K}\mathcal{S}$ and the column space
$\mathcal{R}(B)$. We utilize this isomorphism throughout the paper,
allowing us to compute functional geometric properties---such as
principal angles and orthogonal projections---directly via their
finite-dimensional Euclidean representations.

Specifically, let
$S = \text{span}(\Psi) \subset \Fc \subseteq L_2(\mu_X)$ be the
subspace spanned by the dictionary. Under the map
\eqref{eq:isomorpshism_map},
\begin{equation}
  \Sc \equiv \Rc(A), \quad \Kc \Sc \equiv \Rc(B).
  \label{eq:iso_equiv}
\end{equation}
We can compute the principal angles and vectors between $\Rc(A)$ and
$\Rc(B)$ using the SVD-based procedure from
Theorem~\ref{thm:Golub1973}. As the following result shows, this gives
us the corresponding elements between $\Sc$ and $\Kc \Sc$.

\begin{proposition}[Computation of principal angles and vectors in
  $L_2(\mu_X)$]\label{prop:pa_L2}
  Let $\{A u_i^A\}_{i=1}^{{s}} \subset \Rc(A)$,
  $\{B v_i^B\}_{i=1}^{{s}} \subset \Rc(B)$, and
  $\{\theta_i(\Rc(A), \Rc(B))\}_{i=1}^{{s}}$ be the principal vectors
  and angles between $\Rc(A)$ and $\Rc(B)$, where
  $u_i^A, v_i^B \in \real^s$ are the coefficients of the principal
  vectors in the bases of $\Rc(A)$ and $\Rc(B)$, resp. 
  Then,
  \begin{enumerate}
  \item $\theta_i(\Sc, \Kc \Sc) = \theta_i(\Rc(A), \Rc(B))$, for all
    $i \in [s]$,
  \item $ u_i^{\Sc}(\cdot) = \Psi(\cdot) u_i^A$,
    $v_i^{\Kc \Sc}(\cdot) = \Psi(\cdot) v_i^B$, for all $ i \in [s]$.
  \end{enumerate}
\end{proposition}
\begin{proof}
  Since $\Ec|_{\Sc}$ and $\Ec|_{\Kc \Sc}$ are isomorphisms, the
  principal angles between $\Sc$ and $\Kc \Sc$ are the same as those
  between $\Rc(A)$ and $\Rc(B)$ due to the
  equivalence~\eqref{eq:iso_equiv}. This establishes the first
  claim. For the second claim, it is easy to verify that
  $\Ec(u_i^{\Sc}) = A u_i^A$ and $\Ec(\Kc v_i^{\Kc \Sc}) = B v_i^B$
  and hence, $\{u_i^{\Sc}\}, \{\Kc v_i^{\Kc \Sc}\}$ satisfy
  Definition~\ref{defn:pa_pv} and are the principal vectors between
  $\Sc$ and $\Kc \Sc$.
\end{proof}

The above discussion clarifies how to compute principal arguments for
the specific $L_2(\mu_X)$ inner product used in the data-driven
setting.
Leveraging this finite-dimensional identification, the algorithms
presented below proceed by iteratively pruning the subspace
$\mathcal{S}$ to find a target subspace $\mathcal{S}^*$ that satisfies
the invariance bound~\eqref{eq:problem-bound}.


\subsection{Single-Principal-Vector (SPV) Pruning}
Given the definition of invariance proximity, the most direct
implementation of the pruning strategy is to iteratively remove the
principal vector corresponding to the maximum principal angle.
Conceptually, the SPV algorithm operates on a ``worst-offender''
principle. At every step, we analyze the alignment between the current
subspace and its image under the dynamics. The principal vector
$u_{\max}^{\Sc} \in \Sc$ corresponding to $\theta_{\max}$ highlights
the specific direction that undergoes the most significant rotation
out of the subspace. By explicitly projecting the subspace onto the
orthogonal complement of~$u_{\max}^{\Sc}$, we surgically remove the
dimension most responsible for the invariance error. This yields a new
subspace and the process is repeated until the maximum angle falls
below the desired tolerance $\epsilon$. This is formalized in
Algorithm \ref{alg:naive-pruning}.

\begin{algorithm}[h]
  \caption{\textbf{SPV Pruning}}
  \label{alg:naive-pruning}
  \begin{algorithmic}[1]
    \REQUIRE $\Sc, \Kc \Sc \subset \Fc$, $\epsilon \in [0,1]$
    \STATE Initialize $\Sc_1 \gets \Sc$, \, $i \gets 0$
    \WHILE{True}
    \STATE $i \gets i + 1$
    \IF{$\Sc_{i} = \emptyset$}
    \RETURN $\emptyset$ \RightComment{Terminate with failure}
    \ENDIF
    \STATE $\{u^{\Sc}_j\}, \{\theta_j\} \gets \text{Compute principal
      arguments}(\Sc_i, \Kc \Sc_i)$ 
    \IF{$\sin \theta_{\max} \leq \epsilon$}
    \RETURN $\Sc_i$ \RightComment{Terminate with success}
    \ENDIF
    \STATE $\Sc_{i+1} \gets \Sc_i \setminus \text{span}(u_{\max}^{\Sc})$
    \ENDWHILE
  \end{algorithmic}
\end{algorithm}

\subsection{Equivalence of RFB-EDMD and SPV
  Algorithms}\label{sec:equivalence}

Here, we explain the equivalence of SPV pruning with the Recursive
Forward-Backward EDMD (RFB-EDMD) algorithm introduced
in~\cite{MH-JC:25-access}.  To do so, we start by introducing key
ingredients of RFB-EDMD. Given data matrices
$X, X^+ \in \mathbb{R}^{N \times n}$ from the nonlinear
system~\eqref{eqn:sys_dynamics}, consider the standard ``forward''
EDMD matrix $K_f = \Psi(X)^\dagger \Psi(X^+)$ (corresponding to the
forward-in-time evolution $x \to x^+$) and the ``backward'' EDMD
matrix $K_b = \Psi(X^+)^\dagger \Psi(X)$ (corresponding to the
backward-in-time evolution $x^+ \to x$).  Let $M_c = I - K_f K_b$ be
the \textit{consistency matrix} measuring the discrepancy between the
forward and backward predictions.
The next result establishes that the eigenvalues of the consistency
matrix are exactly the squared sines of the principal angles between
the search space $\mathcal{S}$ and its image
$\Kc \mathcal{S}$.
%
%
Furthermore, the result identifies the eigenvectors of
$M_c$ as the coefficients of the principal vectors in~$\Sc$.

\begin{lemma}\longthmtitle{Spectral Characterization of
    Consistency}\label{lem:spectral-consistency}
  Let $\Sc \subset \Fc$ be the subspace spanned by the
  dictionary~$\Psi = [\psi_1, \dots, \psi_s]$. Let
  $A = \Psi(X) \in \real^{N \times s}$ and
  $B = \Psi(X^+) \in \real^{N \times s}$ be the data matrices
  representing the domain and image of the Koopman operator on
  $\mathcal{S}$, with full column rank. Then, the consistency matrix
  $M_c = I - K_f K_b$ satisfies:
  \begin{enumerate}
  \item Its eigenvalues $\{\lambda_i \}_{i=1}^s$ are squared sines of
    the principal angles, i.e., $ \lambda_i = \sin^2 \theta_i$ 
    $\forall \,\, i \in [s]$;

  \item Its eigenvectors $\{v_i\}_{i=1}^s$ correspond to the principal
    vectors of $\mathcal{S}$, as specified by
    $u^{\Sc}_i(\cdot) = \Psi(\cdot)v_i$, $\forall \,\,i \in [s]$.
  \end{enumerate}
\end{lemma}
\begin{proof}
  Using the definition of the forward $ K_f = A^\dagger B$ and
  backward $ K_b = B^\dagger A $ EDMD matrices, we have that
  $M_c = I - K_f K_b = I - A^\dagger B B^\dagger A$. Note that
  $P_B = B B^\dagger$ is the orthogonal projection onto
  $\mathcal{R}(B)$. Therefore, $M_c = I - A^\dagger P_B A$.  Consider
  the QR decompositions $A = Q_A R_A$ and $B = Q_B R_B$, where
  $Q_A, Q_B$ have orthonormal columns and $R_A, R_B$ are invertible
  upper triangular matrices. Substituting into the expression for
  $M_c$, we get
  \begin{align*}
    M_c &= I - (R_A^{-1} Q_A^{\top}) (Q_B Q_B^{\top}) (Q_A R_A) \\
        &= I - R_A^{-1} (Q_A^{\top} Q_B Q_B^{\top} Q_A) R_A .
  \end{align*}
  The term $Q_A^{\top} Q_B$ is the matrix of inner products between the
  orthonormal bases. Following Theorem~\ref{thm:Golub1973}, let the SVD of
  this matrix be $ Q_A^{\top} Q_B = U_A (\cos \Theta) V_B^{\top} $, where
  $\cos \Theta = \text{diag}(\cos \theta_1, \dots, \cos \theta_s)$
  contains the cosines of the principal angles. Substituting this SVD
  back, we obtain:
  \begin{align*}
    Q_A^{\top} Q_B Q_B^{\top} Q_A
    = U_A
      \cos^2 \Theta U_A^{\top} .
  \end{align*}
  Therefore, the consistency matrix becomes:
  \begin{align*}
    M_c &= I - R_A^{-1} (U_A \cos^2 \Theta U_A^{\top}) R_A
    \\
        &= R_A^{-1} (I - U_A \cos^2 \Theta U_A^{\top}) R_A \\
        &= R_A^{-1} U_A (I - \cos^2 \Theta) U_A^{\top} R_A .
  \end{align*}
  Since $I - \cos^2 \Theta = \sin^2 \Theta$, we have
  \begin{align*}
    M_c = (R_A^{-1} U_A) \sin^2 \Theta (R_A^{-1} U_A)^{-1}.    
  \end{align*}
  This similarity relation proves that the eigenvalues of $M_c$ are
  exactly $\{\sin^2 \theta_i\}_{i=1}^s$. Furthermore, the eigenvectors
  of $M_c$ are the columns of $R_A^{-1} U_A$. The $i$-th eigenvector
  $v_i = R_A^{-1} u^A_i$ of $M_c$ satisfies
  \begin{align*}
    \Psi(X) v_i = A v_i = Q_A R_A (R_A^{-1} u^A_i) = Q_A u^A_i.    
  \end{align*}
  Since $u^A_i$ is the left singular vector of $Q_A^{\top} Q_B$, according
  to Theorem~\ref{thm:Golub1973}, the vector $Q_A u^A_i = A v_i$ is
  precisely the $i$-th principal vector of the subspace
  $\mathcal{R}(A)$. Utilizing Proposition~\ref{prop:pa_L2}, we have
  $u^{\Sc}_i(\cdot) = \Psi(\cdot)v_i$ for all $i \in [s]$.
\end{proof}
\smallskip

We leverage this spectral characterization to establish the algebraic
equivalence between the RFB-EDMD and SPV pruning algorithms.

\begin{theorem}[Algorithmic Equivalence]\label{thm:RFB-EDMD_SPV}
  The \textit{RFB-EDMD} algorithm is algebraically equivalent to the
  \textit{SPV} pruning algorithm. Specifically, at every iteration,
  both algorithms remove the same one-dimensional subspace from the
  search space $\mathcal{S}$.
\end{theorem}
\begin{proof}
  Let $\mathcal{S}_k$ be the search space at iteration $k$:

  The \textit{RFB-EDMD} algorithm computes the consistency matrix
  $M_c$ and identifies the eigenvector $v_{\max}$ corresponding to the
  largest eigenvalue $\lambda_{\max}$. It then updates the subspace to
  be the orthogonal complement of this direction according to
  $\mathcal{S}_{k+1} = \mathcal{S}_k \setminus
  \text{span}\{\Psi(\cdot)v_{\max}\}$.

  The \textit{SPV} algorithm computes the principal vectors between
  $\mathcal{S}_k$ and its image $\Kc \Sc_k$. It identifies the
  principal vector $u_{\max}^{\Sc} \in \Sc_k$ corresponding to the
  largest principal angle $\theta_{\max}$ and removes it according to
  $\mathcal{S}_{k+1} = \mathcal{S}_k \setminus
  \text{span}\{u_{\max}^{\Sc}\}$.
  
  From Lemma~\ref{lem:spectral-consistency},
  $\lambda_{\max} = \sin^2 \theta_{\max}$. Since $\sin^2 \theta$ is
  monotonic $[0, \pi/2]$, maximizing the eigenvalue is equivalent to
  maximizing the principal angle.  Furthermore,
  Lemma~\ref{lem:spectral-consistency} establishes that the function
  defined by the eigenvector $v_{\max}$ is exactly the principal
  vector $u_{\max}^{\Sc}$, i.e.,
  $ \Psi(\cdot)v_{\max} \equiv u_{\max}^{\Sc} $.  Therefore, both
  algorithms remove the exact same function direction from the
  dictionary span at every step, generating the same sequence of
  nested subspaces $\mathcal{S}_0 \supset \mathcal{S}_1 \supset \dots$
  and terminating at the same final subspace.
\end{proof}

\subsection{Multiple Principal Vector (MPV) Pruning}

The equivalence established in Theorem \ref{thm:RFB-EDMD_SPV} provides
a geometric interpretation of RFB-EDMD and brings up the possibility
of further refinement for increased computational efficiency. This is
because the SPV strategy (and, by extension, RFB-EDMD) removes only a
single dimension --the ``worst offender''-- at each iteration. For
high-dimensional dictionaries where many directions might
simultaneously violate the invariance tolerance $\epsilon$, this
sequential approach may be computationally intensive.  This
observation raises a natural question: is it reasonable to prune
\textit{all} violating directions simultaneously? The following result
provides justification for this type of aggressive pruning, showing
that any direction currently violating the tolerance is fundamentally
incompatible with the target invariance condition.

%
%

\begin{lemma}\longthmtitle{Dropping Multiple Principal Vectors}\label{Lemma:dropping_mult_pvs}
  %
  Consider the subspaces $\Sc, \Kc \Sc \subset \Fc$ according to \eqref{eqn:setting_principal_arguments}. Given a fixed
  $\epsilon^{\star} > 0$, let $\theta_k$ be the smallest principal
  angle satisfying $\sin \theta_k > \epsilon^{\star}$.  Define the
  subspace of violating directions
  $\Wc = \text{span}(u^{\Sc}_i \, \big| \, \sin \theta_i > \epsilon^{\star})
  = \text{span}(u^{\Sc}_{i})_{i = k}^{s}$.  Let
  $\Sc_{\epsilon} \subseteq \Sc$ be any subspace satisfying the
  invariance condition $\delta(\Sc_{\epsilon}) \leq \epsilon$, where
  $\epsilon \leq \epsilon^{\star}$. Then,
  \begin{equation}
    \Wc \cap \Sc_{\epsilon} = \left\{0\right\} .
  \end{equation}
  %
  %
\end{lemma}
\begin{proof}
  Suppose, for the sake of contradiction, that there exists a nonzero unit
  vector $u \in \Wc$ such that $u \in \Sc_{\epsilon}$. Utilizing
  Corollary \ref{corollary:bounds_proj_norm}, we can lower bound the
  projection of $u$ onto the image of the target subspace:
  \begin{align}\label{eq:proj_lower_bound}
    \| \Pc_{\Kc \Sc}u \|^2
    &= \| \Pc_{\Kc \Sc_{\epsilon}}u \|^2 + \|
      \Pc_{\Kc \Sc \setminus \Kc \
      Sc_{\epsilon}}u \|^2 \geq \| \Pc_{\Kc
      \Sc_{\epsilon}}u \|^2 \notag
    \\ 
    &\geq \cos^2 \theta_{\max}(\Sc_{\epsilon},
      \Kc \Sc_{\epsilon}) \, \| u \|^2 \geq (1
      - \epsilon^2) \, \| u \|^2.  
  \end{align}
  Alternatively, since $u \in \Wc$, it can be written as
  $u = \sum_{i = k}^{s} c_i u^{\Sc}_i$ for some coefficients
  $c_i \in \real$. We can upper bound the projection using the
  specific principal angles of the violating vectors:
  \begin{align}
    &\max_{u \in \Wc} \frac{\|\Pc_{\Kc \Sc}u\|^2}{\|u\|^2} = 
      \max_{c \in \real^{s-k+1}} \frac{c^{\top} \text{diag}[\cos^2
      \theta_i]_{i = k}^{s} c}{c^{\top} c} = \cos^2 \theta_k \notag
    \\ 
    &\implies \|\Pc_{\Kc \Sc}u\|^2 \leq \cos^2 \theta_k \, \| u \|^2
      < (1 - (\epsilon^{\star})^2) \, \| u \|^2,
      \label{eq:proj_upper_bound}
  \end{align}
  where we have used $u \neq 0$ in the last inequality.  Since
  $\epsilon \leq \epsilon^{\star}$, the lower bound
  \eqref{eq:proj_lower_bound} reads
  $\|\Pc_{\Kc \Sc}u\|^2 \geq (1 - \epsilon^2) \, \| u \|^2 \geq (1 -
  (\epsilon^{\star})^2) \, \| u \|^2$, which contradicts the upper
  bound \eqref{eq:proj_upper_bound}.  Thus,
  $\Wc \cap \Sc_{\epsilon} = \{0\}$.
\end{proof}

Lemma~\ref{Lemma:dropping_mult_pvs} constitutes a strong exclusion
result, establishing that the subspace $\mathcal{W}$, spanned by all
principal vectors exceeding the tolerance, is strictly disjoint from
any admissible subspace $\mathcal{S}_\epsilon$ satisfying the
invariance criterion. It is important to note, however, that
disjointedness does not imply orthogonality; vectors in $\mathcal{W}$
may still possess a non-trivial projection onto
$\mathcal{S}_\epsilon$. Consequently, $\mathcal{S}_\epsilon$ is
generally not contained within the surviving subspace
$\mathcal{S}^{\text{new}} = \mathcal{S} \setminus \mathcal{W}$. This
leads to a leakage of the target subspace $\mathcal{S}_\epsilon$ into
the pruned directions. Furthermore, the more aggressive the pruning
(i.e., the larger the number of principal vectors dropped), the larger
this leakage can be. This phenomenon is analyzed in detail in
Section~\ref{sec:properties_pruning}, where we derive bounds on this
error based on the principal angles of the dropped vectors.
%

Based on Lemma~\ref{Lemma:dropping_mult_pvs}, we propose the
\textbf{Multiple-Principal-Vector (MPV)} pruning strategy. Instead of
pruning dimensions one by one, we discard the entire subspace
$\mathcal{W}$ in a single batch. This modification is formalized in
Algorithm \ref{alg:multi-pv-pruning}.

\begin{algorithm}[h]
  \caption{\textbf{MPV Pruning}}
  \label{alg:multi-pv-pruning}
  \begin{algorithmic}[1]
    \REQUIRE $\Sc, \Kc \Sc \subset \Fc$, $\epsilon \in [0,1]$
    \STATE Initialize $\Sc_1 \gets \Sc$, \, $i \gets 0$
    \WHILE{True}
    \STATE $i \gets i + 1$
    \IF{$\Sc_{i} = \emptyset$}
    \RETURN $\emptyset$ \RightComment{Terminate with failure}
    \ENDIF
    \STATE $\{u^{\Sc}_j\}, \{\theta_j\} = \text{Compute principal arguments}(\Sc_i, \Kc \Sc_i)$
    \IF{$\sin \theta_{\max} \leq \epsilon$}
    \RETURN $\Sc_i$ \RightComment{Terminate with success}
    \ENDIF
    \STATE $\Sc_{i+1} \gets \text{span}\{ u^{\Sc}_j \, | \, \sin \theta_j \leq \epsilon \}$ \RightComment{Batch Pruning}
    \ENDWHILE
  \end{algorithmic}
\end{algorithm}

Note that Algorithm \ref{alg:multi-pv-pruning} differs from Algorithm
\ref{alg:naive-pruning} only in the pruning step (Line 11), yet this
modification can reduce the number of iterations drastically,
especially in high-dimensional settings. As mentioned above, the
downside to MPV pruning is that it can be too aggressive, leading to a
large leakage of the target invariant subspace into the pruned
directions. This can result in a pruned subspace which is much smaller
than what is needed to satisfy the invariance specification.
Furthermore, if there is no subspace that satisfies the invariance
condition, MPV will output the trivial empty subspace. Instead, the
SPV algorithm provides a sequence of nested subspaces that can be used
to trade-off between invariance and expressivity. Thus, if the user is
unsure about what tolerance $\epsilon$ to select, SPV provides a more
gentle approach to gradually reach the desired invariance level.
\begin{remark}[Connection with the Principal Angle Decomposition (PAD)
  algorithm~\cite{GC-NB-JCL-SB-MJC:26}]
  Interestingly, the recently proposed PAD
  algorithm~\cite{GC-NB-JCL-SB-MJC:26} corresponds to the first
  iteration of the MPV algorithm proposed here.  The PAD algorithm
  builds on the computational tools (specifically, the use of Gram
  matrices with entries $ \innerprodF{\Kc\psi_i}{\Kc\psi_j}$,
  $i,j =1,\dots,s$) introduced in~\cite{MJC-LJA-MS:23,MJK:24} for
  ResDMD to compute residuals and subspace-invariance diagnostics.
  The PAD algorithm performs a single batch pruning step, where all
  directions with principal angles above the tolerance are removed
  simultaneously. Note, however, that this does not guarantee that the
  obtained subspace satisfies the invariance
  condition~\eqref{eq:problem-bound}. \oprocend
\end{remark}
%
%

\section{Properties of Subspace Pruning
  Algorithms}\label{sec:properties_pruning}

This section evaluates the properties of the subspace pruning
algorithms, emphasizing the trade-offs between the SPV and MPV
approaches. Crucially, we identify a significant vulnerability in
high-dimensional pruning: the potential for the target invariant
subspace to leak into the pruned directions. By analyzing the
conditions required to mitigate this leakage, we also establish the
foundation for a more robust pruning strategy.

\subsection{Bounds on Information Loss for Subspace
  Pruning}\label{sec:approx-eigen}
%
%

The algorithmic equivalence between the SPV and RFB-EDMD algorithms
established in Section~\ref{sec:equivalence} implies that exact
eigenfunctions contained in the initial dictionary are preserved by
the pruning process, cf.~\cite[Theorem~5]{MH-JC:25-access}. However,
exact invariance is rare, particularly in data-driven settings.  Here,
instead, we are interested in studying the behavior of the pruning
algorithms with respect to $\epsilon$-\emph{approximate}
eigenfunctions, which we define as a function $f \in \Sc$ whose image
$\Kc f$ satisfies
\begin{align*}
  \sin \theta(f,\Kc f) \leq \epsilon,   
\end{align*}
where $\theta(f,\Kc f)$ is the angle between $f$ and $\Kc f$ defined
by
$\cos \theta(f, \Kc f) = \frac{\langle f, \Kc f \rangle}{\|f\| \,
  \|\Kc f\|}$.  We show that the projection of an
$\epsilon$-approximate eigenfunction onto the pruned space remains a
$C\epsilon$-approximate eigenfunction, for some constant $C$ that
depends on the principal angles of the dropped vectors.  To do this,
we examine the core mechanism of the SPV and MPV algorithms, which is
to remove (principal) vectors that least contribute to the invariance,
and quantify the error introduced by this removal.  Intuitively, if we
drop a vector $u_s$ that is far from being invariant (large
corresponding principal angle $\theta_s$), the component of an
approximate function $f$ along $u_s$ must be negligible.

\begin{lemma}[Bound on Information Loss for
  MPV]\label{lemma:pruning_error_bound}
  Consider the subspaces $\Sc, \Kc \Sc \subset \Fc$ according to \eqref{eqn:setting_principal_arguments}. Let
  $\mathcal{I}_{\text{drop}} \subset \{1, \dots, s\}$ be the set of
  indices of the principal vectors to be pruned and
  $\mathcal{S}^{\text{new}} = \text{span}\{u^{\Sc}_i\}_{i \notin
    \mathcal{I}_{\text{drop}}}$ the pruned subspace.  Let
  $\gamma = \min_{k \in \mathcal{I}_{\text{drop}}} \sin \theta_k$ be
  the sine of the minimum principal angle among the dropped
  vectors. If $f \in \mathcal{S}$ is an $\epsilon$-approximate
  eigenfunction satisfying $\| f \| = 1$, then
  \begin{equation}
    \label{eqn:info_loss_bound}
    \text{dist}(f, \mathcal{S}^{\text{new}}) \le \frac{\epsilon}{\gamma}.  
  \end{equation}
\end{lemma}
\begin{proof}
  Let $f = \sum_{i=1}^s c_i u^{\Sc}_i$ be the decomposition of $f$ in the
  principal basis. The distance to the subspace
  $\mathcal{S}^{\text{new}}$ is determined by the energy in the
  dropped components,
  \begin{align*}
    \text{dist}(f, \mathcal{S}^{\text{new}})^2 = \sum_{k \in
    \mathcal{I}_{\text{drop}}} c_k^2 .     
  \end{align*}
  We analyze the projection of $f$ onto the image subspace
  $\Kc \mathcal{S}$. From the assumption
  $\sin \theta(f, \Kc f) \le \epsilon$, we have
  \begin{align*}
    \|P_{\Kc \mathcal{S}}f\|^2 \ge \|P_{\Kc f}f\|^2 \ge 1 - \epsilon^2.
  \end{align*}
  Expanding this in terms of principal angles, we get
  \begin{align*}
    \sum_{i=1}^s c_i^2 \cos^2 \theta_i
    &= \sum_{j \notin
      \mathcal{I}_{\text{drop}}}
      c_j^2 \cos^2 \theta_j +
      \sum_{k \in
      \mathcal{I}_{\text{drop}}}
      c_k^2 \cos^2 \theta_k \ge 1 -
      \epsilon^2. 
  \end{align*}
  We use the bound $\cos^2 \theta_j \le 1$ for the retained components
  and $\cos^2 \theta_k = 1 - \sin^2 \theta_k$ for the dropped
  components:
  \begin{align*}
    \sum_{j \notin \mathcal{I}_{\text{drop}}} c_j^2 + \sum_{k \in
    \mathcal{I}_{\text{drop}}} c_k^2 (1 - \sin^2 \theta_k) &\ge 1 -
                                                             \epsilon^2. 
  \end{align*}
  Grouping the coefficients and utilizing $\sum c_i^2 = 1$:
  \begin{align*}
    1 - \sum_{k \in \mathcal{I}_{\text{drop}}} c_k^2 \sin^2 \theta_k
    &\ge 1 - \epsilon^2 \implies \sum_{k \in
      \mathcal{I}_{\text{drop}}} c_k^2 \sin^2 \theta_k \le \epsilon^2. 
  \end{align*}
  For every dropped vector $k \in \mathcal{I}_{\text{drop}}$, we have
  $\sin \theta_k \ge \gamma$. Thus:
  \begin{align*}
    \gamma^2 \sum_{k \in \mathcal{I}_{\text{drop}}} c_k^2 \le \sum_{k
    \in \mathcal{I}_{\text{drop}}} c_k^2 \sin^2 \theta_k \le
    \epsilon^2. 
  \end{align*}
  Dividing by $\gamma^2$ and taking the square root yields the result.
\end{proof}

\begin{remark}[Bound on Information Loss for
  SPV]\label{re:information-loss-SPV}
  A similar bound holds for the SPV algorithm when dropping a single
  vector $u_s$ with principal sine $\sin \theta_s = \gamma$. In this
  case, the information loss bound simplifies to:
  $ \text{dist}(f, \mathcal{S}^{\text{new}}) \le \frac{\epsilon}{\sin
    \theta_s}$. \oprocend
\end{remark}

Having established that the projected function
$f^{\text{new}} = P_{\mathcal{S}^{\text{new}}}f$ is geometrically
close to $f$, we now prove that it retains the dynamical property of
being an approximate eigenfunction.

\begin{theorem}[Stability of MPV
  Pruning]\label{thm:stability_MPV_pruning}
  With the same notation as Lemma~\ref{lemma:pruning_error_bound}, let
  $f \in \mathcal{S}$ be an $\epsilon$-approximate eigenfunction with
  $\| f \| = 1$ and $\|\Kc f\| \ge m > 0$.  Additionally, suppose
  $\Kc$ restricted to the finite dimensional space $\Sc$ satisfies
  $\| \Kc|_{\Sc} \| \le L$.
  %
  Then, $f^{\text{new}} = P_{\mathcal{S}^{\text{new}}}f$ is a
  $C \epsilon$-approximate eigenfunction, with
  $C$ given by
  $$ C = 1 + \frac{1}{\gamma}\left(2 + \frac{4L}{m}\right). $$
\end{theorem}
\begin{proof}
  From Lemma~\ref{lemma:pruning_error_bound}, we have
  $\|f - f^{\text{new}}\| \le \frac{\epsilon}{\gamma}$.  Using the
  operator bound,
  $\|\Kc f - \Kc f^{\text{new}}\| \le L \|f - f^{\text{new}}\| \le
  \frac{L\epsilon}{\gamma}$.  With normalization lower bounds
  $\alpha_x = \|f\| = 1$ and $\alpha_y = \|\Kc f\| \ge m$, the
  application of Lemma~\ref{lemma:Perturbation_bound} with
  $x=f, y=\Kc f$ and $x'=f^{\text{new}}, y'=\Kc f^{\text{new}}$ yields
  \begin{multline*}
    |\sin \theta(f^{\text{new}}, \Kc f^{\text{new}}) - \sin \theta(f, \Kc f)|
    \le 2 \|f - f^{\text{new}}\| +
    \\
    \frac{4}{m} \|\Kc f - \Kc f^{\text{new}}\| 
    \le \frac{\epsilon}{\gamma} \left( 2 + \frac{4L}{m} \right).
  \end{multline*}
  Using now $\sin \theta(f, \Kc f) \le \epsilon$ yields the desired result.
\end{proof}

For SPV pruning, a similar stability result holds with the same
constant $C$, where $\gamma$ is the principal sine of the single
dropped vector, cf. Remark~\ref{re:information-loss-SPV}. The constant
$C$ in Theorem~\ref{thm:stability_MPV_pruning} acts as a condition
number for the pruning step: it is small when the principal angles of
dropped vectors are well-separated from those of retained vectors
(large $\gamma$) and when the dynamics restricted to the subspace are
well-conditioned (small $L/m$).  Furthermore, the result implies that
exact eigenfunctions ($\epsilon=0$) are exactly preserved during
SPV/MPV pruning. Finally, we iteratively apply
Theorem~\ref{thm:stability_MPV_pruning} to analyze the effect of
multiple pruning steps.
\begin{corollary}[Stability of Multi-Step MPV Pruning]\label{cor:multi_step_pruning}
  Consider a sequence of MPV pruning steps that generate $T$ nested
  subspaces
  $\mathcal{S}_0 \supset \mathcal{S}_1 \supset \dots \supset
  \mathcal{S}_T$ according to
  Algorithm~\ref{alg:multi-pv-pruning}. Let $f_{0} \in \Sc_0$ be an
  $\epsilon$-approximate eigenfunction with $\|f_{0}\| = 1$. At each
  step $k \in [T]$, let $\gamma_k$ be the sine of the minimum
  principal angle among the vectors dropped, and define the
  iteratively pruned function
  $f_{k} = P_{\Sc_k} f_{k-1}/\|f_{k-1}\|$.
  Suppose there exists a uniform lower bound $m > 0$ such that
  $\|\Kc f_{k-1}\| /\|f_{k-1}\| \ge m$ for all $k \in [T]$, and let
  $\| \Kc|_{\Sc_0} \| \le L$. Then the resulting function $f_{T}$ is a
  $C_T \, \epsilon$-approximate eigenfunction, where the accumulated
  stability constant is
  \begin{equation}\label{eqn:multi_step_constant}
    C_T = \prod_{k=1}^{\top} \left[ 1 + \frac{1}{\gamma_k}\left(2 +
        \frac{4L}{m}\right) \right]. 
  \end{equation}
\end{corollary}

\subsection{Retention of External Eigenfunctions}\label{sec:ext-eigen}

The discussion of Section~\ref{sec:approx-eigen} deals with the case
of approximate eigenfunctions that belong to the finite span of the
dictionary. Here, we extend the analysis to the case where true
Koopman eigenfunctions lie slightly outside. This is always the case
in practice because of finite machine-precision. Thus, understanding
how well the pruning algorithms retain these ``external''
eigenfunctions is critical.

Formally, given an given a finite-dimensional subspace
$\Sc \subset \Fc$, an eigenfunction $f \in \Fc$ with eigenvalue
$\lambda \neq 0$ is
\begin{itemize}
\item \emph{$\epsilon$-close} to $\Sc$ with constant $\epsilon>0$ if
  $\text{dist}(f, \Sc) \le \epsilon$;
\item \emph{$\Sc$-admissible} with constant $L_\perp > 0$ if
  $ \|\Kc f^{\perp}\| \le L_\perp \|f^{\perp}\|$, where
  $f^{\perp} = (I - P_\Sc)f$.
\end{itemize}
Intuitively, $\epsilon$-closeness captures the notion of the
eigenfunction being well-approximated by the subspace, whereas
$\Sc$-admissibility captures the bounded amplification of its
orthogonal residual by the dynamics.


The next result shows that eigenfunctions that enjoy both properties
are the ones retained by the pruning algorithm,
highlighting the importance of the selection of the initial
dictionary.

\begin{lemma}[External Eigenfunctions are
  $\epsilon$-Approximate]\label{lemma:external_approximate_eigfunc}
  Let $f \in \Fc$ be an eigenfunction of $\Kc$ with eigenvalue
  $\lambda \neq 0$ satisfying $\|f\| = 1$. Assume $f$ is
  $\epsilon$-close to the subspace $\Sc$, with $\epsilon<1$,
  and $\Sc$-admissible with constant $L_\perp$.  Then,
  $f_{\Sc}=P_{\Sc}f$ is a $C_{\text{ext}} \, \epsilon$-approximate
  eigenfunction of $\Kc$, where
  $$ C_{\text{ext}} = 2 + \frac{4 L_\perp}{|\lambda|}. $$
\end{lemma}
\begin{proof}
  Since orthogonal projections are non-expansive,
  $\|f_{\Sc}\| \le \|f\| = 1$.  Let $f_0^{\perp} = f - f_{0}$ be the
  orthogonal residual. By the $\epsilon$-closeness assumption, we have
  $\|f_{\Sc}^{\perp}\| \le \epsilon$. Because $\epsilon < 1$, we deduce
  $f_{\Sc} \neq 0$.  We invoke Lemma~\ref{lemma:Perturbation_bound} with
  $x = f, y = \Kc f$, $x' = f_{\Sc}$, and $y' = \Kc f_{\Sc}$. The
  normalization bounds are $\alpha_x = \|f\| = 1$ and
  $\alpha_y = \|\Kc f\| = |\lambda|$. We bound the distances between
  the exact and projected components:
  \begin{align*}
    \|x - x'\| &= \|f - f_{\Sc}\| = \|f_{\Sc}^{\perp}\| \le \epsilon,
    \\
    \|y - y'\| &= \|\Kc f - \Kc f_{\Sc}\| = \|\Kc f_{\Sc}^{\perp}\| \le
                 L_\perp \|f_{\Sc}^{\perp}\| \le L_\perp \epsilon.
  \end{align*}
  Applying Lemma~\ref{lemma:Perturbation_bound} yields
  \begin{align*}
    |\sin \theta(f_{\Sc}, \Kc f_{\Sc})| 
    &\le \frac{2}{\alpha_x}\|x - x'\| + \frac{4}{\alpha_y}\|y - y'\|
      = C_{\text{ext}}  \epsilon,
  \end{align*}
  where we have used $\sin \theta(f, \Kc f) = 0$.
\end{proof}


Next, we show that the admissibility of an eigenfunction is inherited
under pruning, meaning that if an eigenfunction is admissible with
respect to the initial dictionary, it remains admissible with respect
to any nested subspace obtained by pruning. This is a crucial property
that ensures the retention of external eigenfunctions across multiple
pruning steps.

\begin{lemma}[Admissibility is preserved under
  Pruning]\label{lemma:hereditary_admissibility}
  Assume the initial subspace $\Sc_{0}$ is such that
  $\|\Kc|_{\Sc_0}\| \le L$. Let $f \in \Fc$ be an $\Sc_0$-admissible
  eigenfunction with constant $L_\perp$. For any nested subspace
  $\Sc_k \subseteq \Sc_0$, $f$ is an $\Sc_k$-admissible eigenfunction
  with constant $\tilde{L}_\perp \le L_\perp + L$.
\end{lemma}
\begin{proof}
  Let $f^{\perp}_0 = (I - P_{\Sc_0})f$ be the residual with respect to
  the initial subspace $\Sc_0$, and $f^{\perp}_k = (I - P_{\Sc_k})f$
  be the residual with respect to the $k$-th subspace $\Sc_k$.  Then,
  \begin{align*}
    f^{\perp}_k
    &= f - P_{\Sc_k}f = (f - P_{\Sc_0}f) + (P_{\Sc_0}f - P_{\Sc_k}f)
    \\
    &= f^{\perp}_0 + \Delta_k,
  \end{align*}
  where $\Delta_k \in \Sc_0$.
  %
  %
  Since $f^{\perp}_0 \in (\Sc_0)^\perp$, the components are
  orthogonal, and we have
  $\|f^{\perp}_k\|^2 = \|f^{\perp}_0\|^2 + \|\Delta_k\|^2$. This
  implies that
  \begin{align}
  \|f^{\perp}_0\| \le \|f^{\perp}_k\| \text{ and } \|\Delta_k\| \le \|f^{\perp}_k\|.
    \label{eq:orthogonal_bounds}
  \end{align}  
  We evaluate the action of the Koopman operator on the new
  residual. By the triangle inequality,
  \begin{equation*}
    \|\Kc f^{\perp}_k\| \le \|\Kc f^{\perp}_0\| + \|\Kc \Delta_k\|.
  \end{equation*}
  For the first term, the $\Sc_0$-admissibility of $f$ guarantees
  $\|\Kc f^{\perp}_0\| \le L_\perp \|f^{\perp}_0\|$.  For the second
  term, since $\Delta_k \in \Sc_0$, we can apply the restricted
  operator bound: $\|\Kc \Delta_k\| \le L \|\Delta_k\|$.  Substituting
  these bounds and utilizing~\eqref{eq:orthogonal_bounds},
  \begin{align*}
    \|\Kc f^{\perp}_k\| &\le L_\perp \|f^{\perp}_0\| + L \|\Delta_k\| 
    \le L_\perp \|f^{\perp}_k\| + L \|f^{\perp}_k\| \\
    &= (L_\perp + L)\|f^{\perp}_k\|.
  \end{align*}
  Therefore, $f$ is $\Sc_k$-admissible with
  $\tilde{L}_\perp \le L_\perp + L$.
\end{proof}

The following result provides a bound on the distance of an external
eigenfunction to the pruned subspace after multiple MPV pruning steps.

\begin{theorem}[External Eigenfunction Bound for Multi-step MPV
  Pruning]\label{thm:tight_distance_bound}
  Consider the setting of Corollary~\ref{cor:multi_step_pruning}.  Let
  $f \in \Fc$ be an eigenfunction of $\Kc$ with eigenvalue
  $\lambda \neq 0$ satisfying $\|f\| = 1$. Assume $f$ is
  $\epsilon$-close to the initial subspace $\Sc_0$, with $\epsilon<1$,
  and $\Sc_0$-admissible with constant $L_\perp$.  Then,
  \begin{equation}
    \text{dist}(f, \Sc_T) \le \epsilon \prod_{k=1}^{\top} \sqrt{ 1 +
      \frac{\tilde{C}_{\text{ext}}^2}{\gamma_k^2} }, 
    \label{eq:Iterative_bound}
  \end{equation}
  where
  $\tilde{C}_{\text{ext}} = 2 + \frac{4(L_\perp + L)}{|\lambda|}$.
\end{theorem}
\begin{proof}
  Let $f_k = P_{\Sc_k} f$ be the orthogonal projection of the exact
  eigenfunction onto the $k$-th subspace, and let
  $\epsilon_k = \|f - f_k\|$. 
  We can decompose $f - f_k$ into
  $ f - f_k = (f - f_{k-1}) + (f_{k-1} - f_k)$, where $f - f_{k-1} \in
  \Sc_{k-1}^\perp$ and $f_{k-1} - f_k \in \Sc_{k-1}$. Therefore,
  \begin{align}\label{eqn:pythagoras_step}
    \epsilon_k^2 = \epsilon_{k-1}^2 + \|f_{k-1} - f_k\|^2.
  \end{align}
  Next, we bound $\|f_{k-1} - f_k\|$.  By
  Lemma~\ref{lemma:hereditary_admissibility}, $f$ is an
  $\Sc_{k-1}$-admissible eigenfunction with
  $\tilde{L}_\perp \le L_\perp + L$.  Applying
  Lemma~\ref{lemma:external_approximate_eigfunc} to $\Sc_{k-1}$, the
  projection $f_{k-1}$ is a
  $(\tilde{C}_{\text{ext}} \epsilon_{k-1})$-approximate eigenfunction,
  where
  $\tilde{C}_{\text{ext}} = 2 + \frac{4(L_\perp + L)}{|\lambda|}$.  We
  now normalize $f_{k-1}$ to apply the pruning error bound.  Since
  orthogonal projections are non-expansive,
  $\|f_{k-1}\| \le \|f\| = 1$.  By
  Lemma~\ref{lemma:pruning_error_bound}, the distance from
  $u = f_{k-1} / \|f_{k-1}\|$ to the next pruned subspace $\Sc_k$ is
  bounded as
  $ \|u - P_{\Sc_k} u\| \leq \frac{\tilde{C}_{\text{ext}}
    \epsilon_{k-1}}{\gamma_k}$. Combining these facts, we deduce
  \begin{align*}
    \|f_{k-1}\!  - \!f_k\| \!=\! \|f_{k-1}\| \!\cdot\! \|u \!-\! P_{\Sc_k} u\| 
    \frac{\tilde{C}_{\text{ext}}
    \epsilon_{k-1}}{\gamma_k}. 
  \end{align*}
  Substituting this bound back into~\eqref{eqn:pythagoras_step}, we
  obtain
  \begin{align*}
    \epsilon_k^2 \le \epsilon_{k-1}^2 + \Big(
    \frac{\tilde{C}_{\text{ext}} \epsilon_{k-1}}{\gamma_k} \Big)^2 =
    \epsilon_{k-1}^2 \Big( 1 +
    \frac{\tilde{C}_{\text{ext}}^2}{\gamma_k^2} \Big). 
  \end{align*}
  Taking the square root of both sides gives the recursive geometric
  bound
  $ \epsilon_k \le \epsilon_{k-1} \sqrt{ 1 +
    \frac{\tilde{C}_{\text{ext}}^2}{\gamma_k^2} }$.  Since
  $\epsilon_0 = \|f - f_0\| \le \epsilon$, unrolling this
  recursion from $k=T$ down to $k=1$ yields the final multiplicative
  bound.
\end{proof}

%


The error bound~\eqref{eq:Iterative_bound} in
Theorem~\ref{thm:tight_distance_bound} reveals a limitation of
subspace pruning: since the amplification factor
$\sqrt{ 1 + \frac{\tilde{C}_{\text{ext}}^2}{\gamma_k^2} }$ at step $k$
is strictly greater than $1$, if the pruning process requires a large
number of iterations $T$, the cumulative error can grow significantly,
indicating a potential loss of spectral information. This effect is
exacerbated if the algorithm begins pruning directions where the gap
$\gamma_k$ is small.

This observation highlights a fundamental trade-off between SPV and
MPV pruning.  \textit{SPV pruning} removes only one direction at a
time, meaning in high-dimensional settings it requires
$T_{\text{SPV}} \approx N_{\text{bad}}$ iterations (where
$N_{\text{bad}}$ is the number of violating dimensions). This large
iteration count risks significant cumulative error, though it
mitigates the single-step amplification by carefully removing only the
worst offender with the largest gap $\gamma_k$. Conversely,
\textit{MPV pruning} removes all violating directions in a single
batch, drastically reducing the number of iterations
($T_{\text{MPV}} \ll T_{\text{SPV}}$). However, MPV can be overly
aggressive; dropping directions with small gaps $\gamma_k$ leads to a
large amplification factor and loss of spectral information.

\begin{algorithm}[tb]
  \caption{MPV-SPV Hybrid Pruning}
  \label{alg:hybrid_pruning}
  \begin{algorithmic}[1]
    \REQUIRE $\Sc, \Kc \Sc \subset \Fc$, $\epsilon \in [0,1]$,
    $\epsilon_{\text{coarse}} > \epsilon$. 
    \smallskip
    \STATE $\Sc_{\text{coarse}} \leftarrow \text{MPV}(\Sc, \Kc \Sc,
    \epsilon_{\text{coarse}})$  
    \STATE $\Sc_{\text{final}} \,\,\, \leftarrow \text{SPV}(\Sc_{\text{coarse}}, \Kc \Sc_{\text{coarse}}, \epsilon)$ 
    \RETURN $\Sc_{\text{final}}$
  \end{algorithmic}
\end{algorithm}

This insight leads us to propose a hybrid strategy, summarized in
Algorithm~\ref{alg:hybrid_pruning}, that balances these trade-offs by
minimizing the number of iterations while maximizing the spectral gap
of pruned directions at each step.  First, we apply the MPV algorithm
with a relaxed tolerance $\epsilon_{\text{coarse}} > \epsilon$. This
acts as a coarse pruning step, rapidly eliminating the most
non-invariant directions. Because MPV operates in batches, it
terminates in very few iterations. The relaxed tolerance ensures that
only vectors with large principal angles are removed, maintaining
relatively large gaps $\gamma_k$ and controlling the single-step error
amplification. This yields a subspace $\mathcal{S}_{\text{coarse}}$ of
significantly smaller dimension without incurring large numerical
drift.

Subsequently, we refine $\mathcal{S}_{\text{coarse}}$ using the SPV
algorithm with the strict target tolerance $\epsilon$. Since the
dimension of $\mathcal{S}_{\text{coarse}}$ is already reduced, SPV
converges rapidly, introducing minimal cumulative numerical
error. This final stage provides a nested sequence of subspaces,
allowing for a precise trade-off between expressivity and
invariance. In Section~\ref{sec:sim-numerical_benchmarking} below, we
empirically validate Algorithm~\ref{alg:hybrid_pruning} and illustrate
its numerical stability.

\section{Efficient Computation of Principal Angles and Vectors via
  Rank-One Updates}\label{sec:efficient_computation}

This section develops an efficient procedure for computing principal
angles and vectors after having pruned multiple principal vectors from
a subspace. Interestingly, this procedure can be applied to arbitrary
inner products. We state our theoretical results in a general inner
product space $\Fc$, but when describing the algorithmic
implementation and time complexity, we focus on the $L_2(\mu)$-inner
product.  The key idea is to cast the update as a sequence of
symmetric rank-one corrections to an eigenproblem, enabling fast
incremental updates. This procedure can be directly integrated into
both SPV and MPV pruning to speed up computations and improve
numerical stability.

\subsection{Computation of Principal
  Arguments}\label{sec:computation_principal_angles}

Consider the subspaces $\Sc, \Kc \Sc \subset \Fc$ according to \eqref{eqn:setting_principal_arguments}. Let
\begin{align*}
  \Uc & = [u^{\Sc}_1 \, u^{\Sc}_2 \, \dots \, u^{\Sc}_s], \quad
        \Kc \Uc = [\Kc u^{\Sc}_1 \, \Kc u^{\Sc}_2 \, \dots \, \Kc
        u^{\Sc}_s],
  \\
  \Lambda_{\cos}
      & \!=\! \text{diag}(\cos \theta_1, \dots, \cos \theta_s)
        , \,
  \Lambda_{\sin} \!=\! \text{diag}(\sin \theta_1, \dots, \sin
  \theta_s). 
\end{align*}
Define the pruned subspace
$\Sc^{\text{new}} = \text{span}(u^{\Sc}_1, u^{\Sc}_2, \dots,
u^{\Sc}_{s-k})$ obtained by dropping the top $k$ principal vectors.
The updated image subspace is
$\Kc \Sc^{\text{new}} = \text{span}(\Kc u^{\Sc}_1, \Kc u^{\Sc}_2,
\dots, \Kc u^{\Sc}_{s-k})$.

%
%
We introduce a sequence of orthogonal vectors
$\{\omega_l\}_{l=1}^{k} \subset \Kc \Sc$ that span the orthogonal
complement of the new image subspace $\Kc \Sc^{\text{new}}$ within
$\Kc \Sc$. This sequence helps us decompose the projection onto the
new image subspace $\Kc \Sc^{\text{new}}$ into a projection onto the
old image subspace $\Kc \Sc$ followed by a sequence of rank-one
updates, which can be computed efficiently.
%
%
The sequence is extracted directly from the thin QR decomposition
\cite[Section 5.2]{GHG-CFVL:13}
%
%
of the image matrix $ \Kc \Uc$,
\begin{equation}
  W R = [\Kc u^{\Sc}_1, \Kc u^{\Sc}_2, \dots, \Kc u^{\Sc}_s],
  \label{eq:QR_W}
\end{equation}
where $R \in \real^{s \times s}$ is an upper triangular matrix and
$W = [w_1 \, w_2 \, \dots \, w_s]$ is a matrix with orthonormal
columns.
%
%
Because the first $s-k$ columns of $\Kc \Uc$ span
$\Kc \Sc^{\text{new}}$, the standard QR algorithm guarantees that the
first $s-k$ columns of $W$ form an orthonormal basis for
$\Kc \Sc^{\text{new}}$. The remaining~$k$ columns span the orthogonal
complement of $\Kc \Sc^{\text{new}}$ within~$\Kc \Sc$.

To process the dropped directions sequentially from the least
invariant to the most invariant,
%
%
we define our update sequence $\{\omega_l\}_{l=1}^k$ by collecting
these last $k$ columns of $W$ in reverse order:
\begin{equation}
  \omega_l = w_{s - l + 1}, \quad \text{for } l \in [k].
\end{equation}
We group these vectors into the update matrix
$\Wc_k = [\omega_1 \, \omega_2 \, \dots \, \omega_k]$. By constructing
$\Wc_k$ via the QR decomposition~\eqref{eq:QR_W}, we avoid explicit
and numerically unstable recursive projections, ensuring a robust
foundation for the rank-one updates.

For $l \in [k]$, let
$\omega_l = \sum_{i=1}^{s} d_i^{\omega_l} \Kc v^{\Kc \Sc}_i$, and
define
$d^{\omega_l} = [d^{\omega_l}_1 \, \dots \,
d^{\omega_l}_s]^{\top}$. By construction, the coordinate vectors are
orthonormal and satisfy
$(d^{\omega_l})^{\top} d^{\omega_l} = \delta_{ll}$.  Define the
sequence of rank-one update matrices
$\{N_l\}_{l=0}^{k} \subset \mathbb{R}^{s \times s}$ by
\begin{align}
  N_0 \!=\! \Lambda_{\sin}^2, \,
  N_l \!=\! N_{l-1} + \Lambda_{\cos} d^{\omega_l} (\Lambda_{\cos}
  d^{\omega_l})^{\top} , \quad l \in [k]. 
    \label{eq:rank_one_matrices}
\end{align}

The following result describes to how to compute the principal angles
and vectors between $\Sc^{\text{new}}$ and~$\Kc \Sc^{\text{new}}$
efficiently using the eigenpairs of~$N_k$.

\begin{theorem}[Efficient Computation of Principal
  Arguments]\label{thm:efficient_computation}
  Let $\tilde{N}_k \in \mathbb{R}^{(s-k) \times (s-k)}$ be the
  truncated matrix obtained by dropping the last $k$ rows and columns
  of $N_k$. Let
  $(\lambda_{\alpha} \in \mathbb{R}, z_{\alpha} \in
  \mathbb{R}^{s-k})_{\alpha = 1}^{s-k}$ be the eigenpairs of
  $\tilde{N}_k$, arranged with increasing eigenvalues.  Then, the
  principal vectors
  $\{u^{\Sc^{\text{new}}}_{\alpha}\}_{\alpha=1}^{s-k} \subset
  \Sc^{\text{new}}$ and squared principal sines
  $\sin^2 \theta_{\alpha}(\Sc^{\text{new}}, \Kc \Sc^{\text{new}})$ are
  given by
  \begin{subequations}
    \begin{align}
      \sin^2 \theta_{\alpha} (\Sc^{\text{new}}, \Kc \Sc^{\text{new}})
      &=
        \lambda_{\alpha}, \label{eq:new_angles}
      \\ 
      u^{\Sc^{\text{new}}}_{\alpha} &= \Uc_{s-k} z_{\alpha}, \label{eq:new_vectors} 
    \end{align}
  \end{subequations}
  for $\alpha \in [s-k]$, where
  $\Uc_{s-k} = [u^{\Sc}_1 \, u^{\Sc}_2 \, \dots \,
  u^{\Sc}_{s-k}]$.
\end{theorem}
\smallskip
\begin{proof}
  Based on the construction of $\{\omega_l\}_{l=1}^k$, we have the
  orthogonal direct sum
  $\Kc \Sc = \Kc \Sc^{\text{new}} \,\oplus\, \text{span}(\omega_1,
  \dots, \omega_k)$. Thus, the orthogonal projection operator can be
  expressed as
  \[
    \mathcal{P}_{\Kc \Sc^{\text{new}}} = \mathcal{P}_{\Kc \Sc} -
    \sum_{l=1}^{k} \mathcal{P}_{\omega_l}.
  \]
  Let $\tilde{u} \in \Sc^{\text{new}}$ be decomposed as
  $\tilde{u} = \sum_{\alpha=1}^{s-k} c_\alpha u^{\Sc}_{\alpha}$.
  %
  %
  The projection of $\tilde{u}$ onto the original image
  space $\Kc \Sc$ yields
  \[
    \mathcal{P}_{\Kc \Sc} \tilde{u} = \sum_{\alpha=1}^{s-k} c_\alpha \Kc v^{\Kc
      \Sc}_{\alpha} \cos \theta_{\alpha}.
  \]
  To compute the projection onto $\omega_l$, we let
  $t_l = \langle \omega_l, \tilde{u} \rangle_{\Fc}$. Then,
  \begin{align*}
    t_l &= \Big\langle \sum_{i=1}^{s} d^{\omega_l}_i \Kc v^{\Kc \Sc}_i,
    \sum_{\alpha=1}^{s-k} c_\alpha u^{\Sc}_{\alpha} \Big\rangle_{\Fc} \\
    &= \sum_{\alpha=1}^{s-k} d^{\omega_l}_\alpha c_\alpha \cos
      \theta_{\alpha} = \tilde{c}^{\top} 
    \Lambda_{\cos} d^{\omega_l}, 
  \end{align*}
  where $\tilde{c} = [c_1, \dots, c_{s-k}, 0, \dots, 0]^{\top} \in \mathbb{R}^s$.
  The projection is therefore
  $\mathcal{P}_{\Kc \Sc^{\text{new}}} \tilde{u} = \mathcal{P}_{\Kc
    \Sc} \tilde{u} - \sum_{l=1}^{k} t_l \omega_l$.
  %
  Since the basis elements are orthonormal, the squared norm
  evaluates as
  \begin{align*}
    \| \mathcal{P}_{\Kc \Sc^{\text{new}}} \tilde{u} \|_{\Fc}^2 
    &=   
      \big\| \Lambda_{\cos} \tilde{c} - \sum_{l=1}^{k} t_l d^{\omega_l} \big\|_2^2
    \\
    &= \tilde{c}^{\top} \Lambda_{\cos}^2 \tilde{c} + \sum_{l=1}^{k} t_l^2 - 2
      \sum_{l=1}^{k} t_l \tilde{c}^{\top} \Lambda_{\cos} d^{\omega_l}
    \\
    &= \tilde{c}^{\top} \Big( \Lambda_{\cos}^2 - \sum_{l=1}^{k} \Lambda_{\cos}
      d^{\omega_l} (\Lambda_{\cos} d^{\omega_l})^{\top}  \Big) \tilde{c}.  
  \end{align*}
  Using equation \eqref{eq:rank_one_matrices}, this expression can be
  rewritten as
  $\| \mathcal{P}_{\Kc \Sc^{\text{new}}} \tilde{u} \|_{\Fc}^2 =
  \tilde{c}^{\top} (I - \tilde{N}_k) \tilde{c}$. Since
  $\| \tilde{u} \|_{\Fc}^2 = \tilde{c}^{\top} \tilde{c}$, utilizing
  Lemma~\ref{lemma:cos_theta_max}, we find the principal angles by
  solving the Rayleigh quotient minimization:
  \[
    \min_{\tilde{c} \in \mathbb{R}^{s-k}}
    \frac{\tilde{c}^{\top} (I - \tilde{N}_k) \tilde{c}}{\tilde{c}^{\top}
      \tilde{c}}.
  \]
  %
  %

  Note that Lemma~\ref{lemma:cos_theta_max} dictates that subsequent
  minimizers must be orthogonal in $\Fc$ (i.e.,
  $\langle \tilde{u}_a, \tilde{u}_b \rangle_{\Fc} =
  0$). Orthonormality of the basis $\Uc^{\text{new}}$ guarantees this
  is equivalent to the Euclidean constraint
  $\tilde{c}_a^\top \tilde{c}_b = 0$. By the Courant-Fischer min-max
  theorem \cite{GHG-CFVL:13},
  sequentially minimizing a Rayleigh quotient subject to
  Euclidean orthogonality constraints is exactly equivalent to
  computing the eigendecomposition of the symmetric matrix
  $(I-\tilde{N}_{k})$.
  
  The successive minimums of this quotient yield the eigenvalues
  $\cos^2 \theta_\alpha$. Consequently, the eigenvalues of
  $\tilde{N}_k$ are exactly
  $1 - \cos^2 \theta_\alpha = \sin^2 \theta_\alpha$, verifying
  equation \eqref{eq:new_angles}. Furthermore, the corresponding
  eigenvectors $z_\alpha$ directly provide the parameterization for
  the updated principal vectors
  $u_{\alpha}^{\Sc^{\text{new}}} = \Uc^{\text{new}}z_\alpha$, yielding
  equation \eqref{eq:new_vectors}.
\end{proof}

\begin{remark}[Applicability to General Inner Product
  Spaces]\label{rem:general_spaces}
  It is important to emphasize that this rank-one update procedure is
  entirely coordinate free and holds for general abstract inner
  product spaces, 
  not just Euclidean spaces. Because the formulation relies purely on
  the principal angles and the coefficients of the QR decomposition,
  the algorithm never requires evaluating the abstract inner product
  $\langle \cdot, \cdot \rangle_{\Fc}$ explicitly during the update
  steps. \oprocend
\end{remark}

\begin{remark}[Computational Complexity and LAPACK]\label{rem:lapack}
  The symmetric rank-one update problem involves computing the
  eigendecomposition of $A \pm \rho \, u u^{\top}$ where the
  eigendecomposition of the symmetric matrix
  $A \in \mathbb{R}^{n \times n}$ is already known. The computational
  cost of this update is $O(n^2)$, implemented robustly in the LAPACK
  subroutine \texttt{DLAED9}~\cite{EA-LAPACK:99}. 
  %
  In the context of Theorem~\ref{thm:efficient_computation}, this
  procedure applies directly to the sequence of matrices
  $\{N_l\}_{l=0}^k$. By chaining $k$ such rank-one updates to compute
  the eigenpairs of $N_k$ (and subsequently the truncated matrix
  $\tilde{N}_k$), the overall cost of our pruning step is $O(k
  s^2)$. When $k \ll s$, this is significantly lower than the $O(s^3)$
  cost of recomputing the principal angles. \oprocend
\end{remark}

\subsection{Incremental Basis Update via QR Decomposition}

While rank-one updates efficiently yield the new principal angles,
iterative pruning also requires maintaining an orthogonal basis for
the image space $\Kc \Sc$. If we were to use the empirical $L_2(\mu)$
inner product over $N$ data points, a naive re-computation of the thin
%
%
QR decomposition for the image matrix $\Kc \Uc$ as in~\eqref{eq:QR_W}
would incur a prohibitive computational cost of $O(Ns^2)$ at each
step.

To circumvent this, we update the factors incrementally. Suppose we
have the decomposition $\Kc \Uc = W R$ available. By restricting
operations to the smaller triangular matrix $R$, we avoid processing
the full, high-dimensional matrix $\Kc \Uc$ directly. We propose the
following efficient procedure to compute the QR decomposition of the
updated image $\Kc \Uc^{\text{new}}$.  Recall from
Theorem~\ref{thm:efficient_computation} that the new principal vectors
$\Uc^{\text{new}}$ are formed by taking linear combinations of the
retained basis using the computed eigenvectors, i.e.,
\begin{equation}
  \Uc^{\text{new}} = [u^{\Sc^{\text{new}}}_1 \, \dots \,
  u^{\Sc^{\text{new}}}_{s-k}] = \underbrace{[u^{\Sc}_1 \, \dots \,
    u^{\Sc}_{s-k}]}_{\text{Retained Basis}} \underbrace{[z_1 \, \dots
    \, z_{s-k}]}_{E}. 
\end{equation}

\textbf{1. Construct the Re-alignment Matrix}\\
Let $E \in \mathbb{R}^{(s-k) \times (s-k)}$ be the matrix of
eigenvectors. We construct a transformation matrix
$T \in \mathbb{R}^{s \times (s-k)}$ by padding $E$ with zeros to align
with the original $s$-dimensional space:
\begin{equation}
  T = \begin{bmatrix} E \\ 0_{k \times (s-k)} \end{bmatrix}.
\end{equation}

\textbf{2. Update the Triangular Factor}\\
We apply the transformation $T$ to the existing upper triangular
factor $R$ to form the intermediate matrix
$C = R T \in \mathbb{R}^{s \times (s-k)}$. We then perform a QR
decomposition of $C$ as
\begin{equation}
  C = Q_C R_C ,
\end{equation}
where $Q_C \in \mathbb{R}^{s \times (s-k)}$ is orthogonal and
$R^{\text{new}} = R_C \in \mathbb{R}^{(s-k) \times (s-k)}$ is the new
upper triangular factor.

\textbf{3. Update the Orthogonal Bases}\\
Finally, we update the orthogonal image basis $W$ by applying the
rotations derived above,
\begin{align}
    W^{\text{new}} &= W Q_C.
\end{align}
As we show next, the resulting matrices form the QR decomposition of
the new image space $\Kc \Uc^{\text{new}}$

\begin{lemma}[Correctness of Incremental
  QR]\label{lemma:incremental_validity}
  %
  %
  Consider the notation and construction of
  Section~\ref{sec:computation_principal_angles}.  The matrices
  $W^{\text{new}}$ and $R^{\text{new}}$ are a valid QR decomposition
  of~$\Kc \Uc^{\text{new}}$, i.e.,
  $\Kc \Uc^{\text{new}} = W^{\text{new}} R_C$.
\end{lemma}
\begin{proof}
  By equation \eqref{eq:new_vectors}, we have
  $\Uc^{\text{new}} = \Uc T$.
  %
  %
  Linearity of the operator $\Kc$ implies
  $\Kc \Uc^{\text{new}} = \Kc \Uc T$. Substituting the initial QR
  decomposition $\Kc \Uc = W R$
  %
  %
  yields
  \[
    \Kc \Uc^{\text{new}} = W R T.
  \]
  Using the definition of $C$ and its decomposition $C = Q_C R_C$, we
  expand the expression as
  \[
    \Kc \Uc^{\text{new}} = W (Q_C R_C) = (W Q_C) R_C = W^{\text{new}} R^{\text{new}}.
  \]
  Since $W$ has orthonormal columns and $Q_C$ is orthogonal, their
  product $W^{\text{new}}$ also has orthonormal columns. Furthermore,
  $R^{\text{new}} = R_C$ is upper triangular by construction. Thus,
  $W^{\text{new}} R^{\text{new}}$ is a valid QR decomposition.
\end{proof}

Based on Lemma~\ref{lemma:incremental_validity}, instead of computing
the QR decomposition of $\Kc \Uc^{\text{new}}$ from scratch (which
costs $O(N(s-k)^2)$ for $N$ data points), we can compute it using the
existing QR decomposition of $\Kc \Uc$ and the QR decomposition of the
smaller matrix $C$, which costs only $O(s(s-k)^2)$. This incremental
update leads to massive computational savings in data-driven
applications where the number of data points vastly exceeds the
dictionary size ($N \gg s$).

%
%
\begin{algorithm}[h]
  \caption{\textbf{Efficient Computation of Principal Arguments}}
  \label{alg:efficient_computation}
  \begin{algorithmic}[1]
    \REQUIRE 
    $\Uc = [u_1, \dots, u_s]$ \RightComment{\textit{PV}s of $\Sc$}
    \REQUIRE 
    $(W ,R) $ \RightComment{Thin QR of $\Kc \Uc$} 
    \REQUIRE 
    $\Lambda_{\sin \theta} \in \real^{s \times s}$
    \RightComment{Principal sines of $\Sc, \Kc \Sc$}  
    \REQUIRE 
    Drop count $k \in \mathbb{N}$ \RightComment{Number of \textit{PV}s to drop}
    \vspace{0.2cm}
    \STATE $\Wc_k \gets [w'_s,\dots,  w'_{s-k+1}]$
    \RightComment{Extract last $k$ columns of $W$} 
    \vspace{0.2cm}
    \STATE $D_W^{\text{cos}} \gets \langle \Uc, \Wc_k \rangle_{\Fc}$ \RightComment{Update vectors $\{d^{\omega_i} \cos \theta_i\}$}
    \vspace{0.2cm}
    \STATE Initialize $\Lambda_0 \gets (\Lambda_{\sin \theta}^2)_{1:s-k, \, 1:s-k}$ \RightComment{Top-left block}
    \STATE Initialize $E_0 \gets I_{s-k}$
    \FOR{$i = 1$ to $k$}
    \STATE $b_i \gets \text{first } s-k \text{ elements of } D_W^{\text{cos}}(:, i)$ 
    \STATE $(\Lambda_i, E_i) \gets \texttt{DLAED9}(\Lambda_{i-1}, E_{i-1}, b_i)$ \RightComment{LAPACK routine} 
    \ENDFOR
    
    \STATE Set $T \gets
    \begin{bmatrix}
      E_k \\ 0
    \end{bmatrix},\, C \gets R T$, $\Lambda_{\sin \theta}^{\text{new}} \gets \Lambda_k^{1/2}, \quad \Uc^{\text{new}} \gets \Uc T$
    \STATE $(Q_C, R_C) \gets \text{QR}(C)$, $W^{\text{new}} \gets W Q_C, \quad R^{\text{new}} \gets R_C$
    \smallskip
    \RETURN $(\Lambda_{\sin \theta}^{\text{new}}, \Uc^{\text{new}}, W^{\text{new}}, R^{\text{new}})$
  \end{algorithmic}
\end{algorithm}


\begin{remark}\longthmtitle{Efficient Algorithm for recomputing
    Principal Arguments}
  Algorithm~\ref{alg:efficient_computation} exploits the results of
  Theorem~\ref{thm:efficient_computation} and
  Lemma~\ref{lemma:incremental_validity} to compute efficiently new
  principal angles and vectors after having dropped multiple principal
  vectors from a subspace.  This procedure can be used as a high-speed
  subroutine in both SPV and MPV pruning.  The algorithm takes as
  input the principal vectors $\Uc$ of the subspace $\Sc$, the
  principal sines $\Lambda_{\sin \theta}$ between $\Sc$ and its image
  $\Kc \Sc$, the QR decomposition $(W,R)$ of $\Kc \Uc$, and the number
  of principal vectors $k$ to drop. The algorithm returns the
  corresponding quantities for the updated subspace $\Sc^{\text{new}}$
  after dropping the top $k$ principal vectors.

  In step 2, we compute the projection matrix $D_W^{\text{cos}}$,
  whose columns are $d^{w_i} \cos \theta_i$. This is evaluated
  directly via the inner product $\langle \Uc, \Wc_k \rangle_{\Fc}$
  (which reduces to the matrix multiplication $\Uc^T \Wc_k$ when
  utilizing the empirical $L_2$ inner product). This avoids the need
  to explicitly compute the principal vectors of~$\Kc \Sc$.  In steps
  3 and 6, we restrict the matrices and vectors to their first $s-k$
  dimensions. This truncation corresponds to finding the eigenpairs of
  $\tilde{N}_k$, which is obtained by dropping the last $k$ rows and
  columns of the full update matrix $N_k$. In step 7, we use the
  LAPACK subroutine \texttt{DLAED9} \cite{EA-LAPACK:99} to compute the
  eigenpairs $(\Lambda_i, E_i)$ of the matrix
  $E_{i-1} \Lambda_{i-1} E_{i-1}^T + b_ib_i^T$ obtained via the
  symmetric rank-one update.
  Section~\ref{sec:sim-computational_efficiency_numerics} describes
  numerical benchmarks demonstrating the vast computational efficiency
  Algorithm~\ref{alg:efficient_computation} when integrated into the
  pruning procedures.  \oprocend
\end{remark}
\vspace{-0.15cm}
\section{Koopman Model for State
  Prediction}\label{sec:state_construction}
\vspace{-0.15cm}
In this section, we outline a procedure for constructing a lifted
linear model for state prediction using the subspace obtained from the
pruning algorithms.
%
Let $s = \dim(\mathcal{S})$ be the dimension of the pruned subspace
obtained after solving Problem \ref{problem:subspace_search} using
either of
Algorithms~\ref{alg:naive-pruning},~\ref{alg:multi-pv-pruning},
or~\ref{alg:hybrid_pruning}. The evolution of observables within this
subspace is governed by the matrix
$K_{\text{EDMD}} \in \mathbb{R}^{s \times s}$, which is the matrix
representation of the projected Koopman operator
$\Pc_{\Sc} \Kc|_{\Sc}$, cf. Section~\ref{sec:EDMD}.

\subsection{Construction of the Lifted Linear Model}
Predicting the evolution of the full state $x \in \real^n$ translates
into describing the Koopman image of the state observables
$e_j(x) = x_j$ for $j \in [n]$ using the pruned subspace $\Sc$. A
natural way to achieve this is to define the approximate Koopman
operator $\hat{\Kc} : \Fc \to \Sc$ as:
\begin{equation*}
  \hat{\Kc} := \Pc_{\Sc} \Kc \circ \Pc_{\Sc}.  
\end{equation*}
This operator projects any observable onto the subspace $\Sc$, applies
the Koopman operator, and then projects the result back onto $\Sc$.
Since $e_j$ does not belong to $\Sc$ in general, we approximate it by
its projection onto $\Sc$. Formally, if
$\Psi = [\psi_1, \dots, \psi_s]^{\top}$ is a basis of $\Sc$,
$\text{span}(\Psi) = \Sc$, then
$ \Pc_{\Sc} e_j = \sum_{i=1}^{s} C_{ji} \psi_i$, which implies
\begin{align*}
  \hat{\Kc} e_j &= \sum_{i=1}^{s} C_{ji} \Pc_{\Sc} \Kc \psi_i.
\end{align*}
Collecting the projection coefficients into a matrix
$C \in \real^{n \times s}$, and utilizing the EDMD matrix
$K_{\text{EDMD}}$ to compute the action of $\Pc_{\Sc} \Kc$ on
$\psi_i$, we obtain the following approximation for the next state:
\begin{equation}
  x^+_j \approx \hat{\Kc} e_j(x) =  c_j^{\top} K_{\text{EDMD}} \Psi(x),
\end{equation}
where $c_j^{\top} \in \real^{1 \times s}$ is the $j$-th row of the
matrix $C$.

By defining the lifted state as $z = \Psi(x) \in \real^{s}$, this
yields the discrete-time lifted linear model:
\begin{align}\label{eq:lifted_dynamics2}
  z^+ = A z, 
  \,\,
  \hat{x} = C z,
\end{align}  
where $A = K_{\text{EDMD}} \in \real^{s \times s}$ is the dynamics
matrix, and $C \in \real^{n \times s}$ is the state reconstruction
matrix obtained by projecting the state observables onto the
basis~$\Psi$.
%
%

\subsection{Balancing Invariance and Reconstruction
  Error}\label{sec:state_reconstruction}

The framework of pruning subspaces while relying on the orthogonal
projection to such subspaces for state reconstruction exposes another
fundamental trade-off. The pruning algorithms generate a nested
sequence of subspaces
$\mathcal{S}_0 \supset \mathcal{S}_1 \supset \dots \supset
\mathcal{S}_T$. As the iteration count $k$ increases, the subspace
dimension decreases, and two competing effects emerge:
\begin{enumerate}
\item \textbf{Invariance Improves:} By design, the invariance
  proximity $\delta(\mathcal{S}_k)$ decreases with~$k$. This ensures
  the accurate long-term evolution of the lifted state $z_k$ under the
  linear dynamics~\eqref{eq:lifted_dynamics2};
\item \textbf{Reconstruction Degrades:} As $\mathcal{S}_k$ shrinks,
  the subspace loses expressive power, and the projection error
  $\| (I - P_{\mathcal{S}_k}) e \|$ of the state observables typically
  increases.
\end{enumerate}
%
%

The designer must carefully weigh the cost of cumulative lifted state
errors against static projection errors. An important observation is
that lifted state errors, which are driven by non-invariance,
accumulate exponentially over time (as the matrix $A$ is raised to
higher powers), whereas the reconstruction error, which are driven by
the state projection onto the subspace (encoded in the matrix $C$) is
static. Therefore, for long-horizon prediction or planning, a tighter
invariant subspace (smaller $\delta$) is preferable. However, for
short-horizon control tasks, a larger subspace with richer state
reconstruction capabilities may be
favored. Section~\ref{sec:sim-vdp_simulation} provides an empirical
example illustrating this trade-off in practice.

\section{Simulation Results}
In this section, we present numerical experiments to validate the
proposed pruning algorithms and demonstrate their effectiveness in
identifying Koopman invariant subspaces. All simulations were
performed on a machine with an Apple M1 Pro chip and 16GB of RAM,
using Python 3.11.4. Unless otherwise specified, we use the standard
inner product on $L_2(\mu_X)$, where $X$ is the trajectory data used
for pruning.

\subsection{Numerical Benchmarking of Pruning
  Algorithms}\label{sec:sim-numerical_benchmarking}

We consider a 2D nonlinear dynamical system to benchmark the numerical
performance of the proposed pruning algorithms. The system dynamics
are given by:
\begin{subequations}\label{eq:sys1}
  \begin{align}
    x_{1}^+ &= 0.8 x_{1} \\
    x_{2}^+ &= \sqrt{0.9 x_{2}^2 + x_{1} + 0.1}
  \end{align}
\end{subequations}
For this system, the Koopman operator possesses 4 known eigenfunctions
among infinitely many others, which are explicitly given by:
$\phi_1(x) = 1$, $\phi_2(x) = x_1$, $\phi_3(x) = x_1^2$, and
$\phi_4(x) = 1 - 10 x_1 - x_2^2$.
%
%
The corresponding eigenvalues are $\lambda_1 = 1$, $\lambda_2 = 0.8$,
$\lambda_3 = 0.64$, and $\lambda_4 = 0.9$. We employ these
eigenfunctions as the ground truth for evaluating the accuracy of the
pruning algorithms.

We initialize the search space with a large dictionary of basis
functions with $1615$ elements, comprising both polynomials (up to
degree 4) and radial basis functions. We generate trajectory data by simulating 500 trajectories, each consisting of 100
time steps, from random initial conditions on $[0,2]^2$.

We apply the pruning algorithms to identify a low-dimensional
approximately invariant subspace that accurately captures the system's
dynamics. It must be noted that the large initial dictionary does
contain all 4 true eigenfunctions, but also includes many irrelevant
basis functions that can introduce numerical drift during pruning. In
the ideal case, the pruning algorithms should be able to recover a
subspace that contains the 4 true eigenfunctions.

Since the initial dictionary is quite large, there are a lot of
linearly dependent basis functions. To mitigate this, we first
performed a rank-revealing QR decomposition on the initial dictionary
to obtain a well-conditioned basis of dimension $s = 274$ before
applying the pruning algorithms. We use
$\epsilon_{\text{tol}} = 10^{-3}$ as the target tolerance for all
pruning algorithms. We implement the rank-one update scheme (Algorithm
\ref{alg:efficient_computation}) within each pruning algorithm to
speed up computations. We also utilize
$\epsilon_{\text{coarse}} = 0.1$ for the MPV pre-pruning step in the
\texttt{MPV\_SPV} algorithm.

\begin{figure}[htb!]
  \centering
  \begin{subfigure}[b]{0.7\linewidth}
    \centering
    \includegraphics[width=\linewidth]{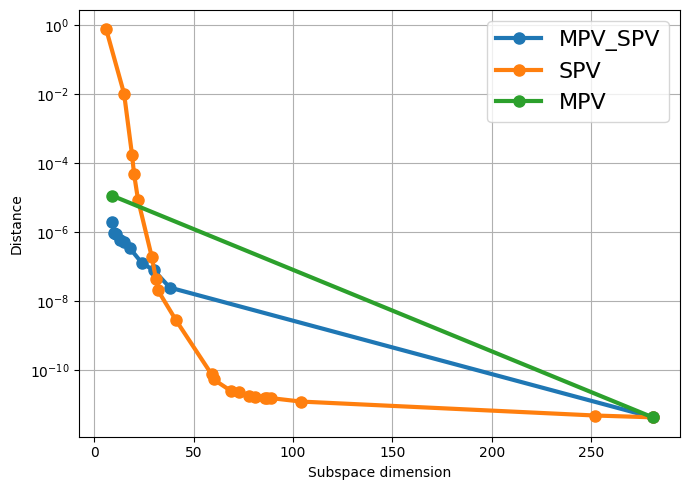}
    \caption{Maximum distance to eigenfunctions}
    \label{fig:accuracy_left}
  \end{subfigure}
  \\
  \begin{subfigure}[b]{0.7\linewidth}
    \centering
    \includegraphics[width=\linewidth]{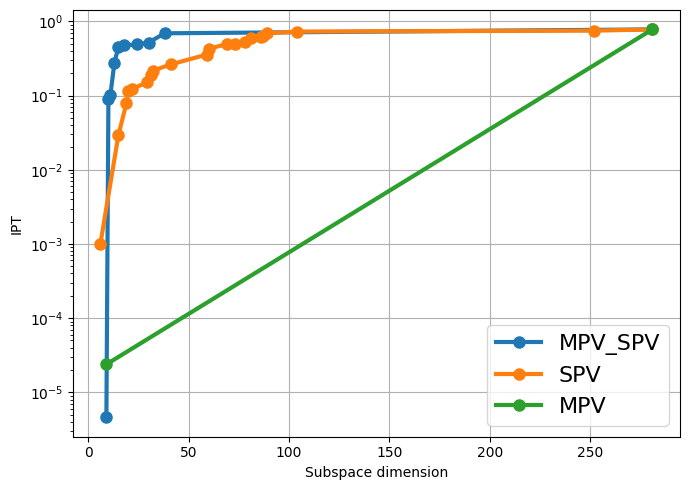}
    \caption{Invariance proximity (IPT)}
    \label{fig:ipt_right}
  \end{subfigure}
  \caption{Comparison of pruning algorithms (read them from right to left). (a) displays
    the maximum distance (Frobenius norm) to the 4 true eigenfunctions
    as a function of the subspace dimension. (b) shows the
    invariance proximity (IPT) of the sequence of pruned subspaces. We
    see that \texttt{MPV\_SPV} (blue) outperforms both \texttt{SPV}
    (orange) and \texttt{MPV} (green).}
  \label{fig:pruning_results}
\end{figure}
%
%

The results are summarized in Figure~\ref{fig:pruning_results}. They
highlight an important drawback of naive iterative pruning methods. If
the search space is too large, numerical drift can lead to significant
degradation in the accuracy of the identified eigenfunctions. We see a
clear hierarchy in performance among the three algorithms. The
\texttt{MPV\_SPV} algorithm, which combines the strengths of both
\texttt{MPV} and \texttt{SPV} methods, is able to mitigate numerical
drift effectively. Furthermore, we also obtain a sequence of nested
subspaces with small invariance proximity which is missing for the
\texttt{MPV} algorithm. This is useful when balancing invariance and
state reconstruction (Section~\ref{sec:state_reconstruction}).

\vspace{-0.2cm}
\subsection{Evaluation of Computational
  Efficiency}\label{sec:sim-computational_efficiency_numerics}

To assess the efficiency gains of the rank-one update scheme,
cf. Algorithm \ref{alg:efficient_computation}, we benchmark the
computation time of the rank-one update scheme against the naive
approach, which recomputes the principal vectors and angles from
scratch at every pruning step. Using the dynamics described in
\eqref{eq:sys1}, we evaluate dictionary sizes of
$s = \{53, 128, 428, 928\}$, constructed via polynomial and radial
basis functions. We also record the runtime of the first SVD
computation to provide a baseline for initialization costs.

The results, summarized in Table \ref{tab:timing_comparison},
demonstrate that the rank-one update scheme yields substantial time
reductions across all algorithms, particularly for larger dimensions
of the initial dimension~$s$. Crucially, in addition to the superior
performance of the \texttt{MPV-SPV} algorithm observed in
Section~\ref{sec:sim-numerical_benchmarking}, the results here confirm
it also maintains a competitive runtime profile when implemented with
rank-one updates. This balance of accuracy and efficiency positions
the rank-one \texttt{MPV-SPV} as a highly practical choice for
large-scale systems. Notably, for large $s$, the total time required
for the pruning updates is comparable to the initial SVD computation
alone, indicating that the overhead of the proposed pruning mechanism
is minimal compared to system initialization.

\begin{table}[ht]
  \centering
  \setlength{\tabcolsep}{2.5pt} 
  
  \resizebox{\columnwidth}{!}{
    \begin{tabular}{rcccccc}
      \toprule
      \textbf{\shortstack{Init.\\Dim.}} &
        \textbf{SPV} & 
        \textbf{\shortstack{SPV\\(Rank-1)}} & 
        \textbf{MPV} & 
        \textbf{\shortstack{MPV\\(Rank-1)}} & 
        \textbf{\shortstack{MPV-SPV\\(Rank-1)}} & 
        \textbf{\shortstack{First\\Comp.}} \\
        \midrule
        53  & 4.0982   & 0.7420  & 0.5486  & 0.5343  & \textbf{0.6296} & 0.2728 \\
        128 & 32.6461  & 2.3755  & 1.3348  & 1.2111  & \textbf{1.3946} & 1.1026 \\
        428 & 176.1153 & 8.7307  & 4.9256  & 3.5185  & \textbf{3.8649} & 3.3766 \\
        928 & 204.0399 & 14.5440 & 10.9164 & 10.2500 & \textbf{9.2133} & 8.8965 \\
      \bottomrule
    \end{tabular}
  }
  \caption{Wall-clock time comparison (in seconds) demonstrating the
    efficiency of rank-one updates across varying dictionary sizes.}
  \label{tab:timing_comparison}
\end{table}

\subsection{Van Der Pol Oscillator}\label{sec:sim-vdp_simulation}
We now demonstrate the application of the proposed pruning algorithms
to construct a Koopman-based model for the Van Der Pol oscillator, a
classic nonlinear dynamical system.
The (discretized) dynamics are given by:
\begin{subequations}\label{eq:vdp}
  \begin{align}
    x_{1}^+ &= x_{1} + \Delta_t \,  x_{2} \\
    x_{2}^+ &= x_{2} + \Delta_t \, \left( (1 - x_{1}^2) x_{2} - x_{1} \right)
  \end{align}
\end{subequations}
where $\Delta_t = 0.025$ is the time step. We generate a dataset of
200,000 snapshot pairs by simulating 500 trajectories, each consisting
of 400 time steps, from random initial conditions uniformly sampled in
the range $[-4, 4]^2$.

The initial dictionary is constructed using the Wendland kernel
\cite{FK-FMP-MS-AS-KW:25} with compact support using centers placed on
a uniform grid over the domain $[-4, 4]^2$ with spacing
$\delta = 0.5$. This leads to a total of $s = 289$ basis functions. As
a baseline comparison, we also implement the kernel EDMD
\cite{FK-FMP-MS-AS-KW:25} using the same Wendland kernel and
centers. Note that kernel EDMD performs the orthogonal projection
using the kernel inner product, which is different from the standard
$L_2(\mu_X)$ inner product used in our other examples. As another
baseline, we perform standard EDMD using the full initial dictionary
without pruning. Finally, we apply the proposed \texttt{MPV-SPV}
pruning algorithm to identify a low-dimensional approximately
invariant subspace working with the standard $L_2(\mu_X)$ inner
product.
We use $\epsilon_{\text{coarse}} = 0.1$ for the \texttt{MPV} step and
$\epsilon_{\text{tol}} = 0.01$ for the \texttt{SPV} step. This process
significantly reduces the dictionary dimension, as shown in
Figure~\ref{fig:vdp_pruning}.

\begin{figure}[ht]
  \centering
  \includegraphics[width=\linewidth]{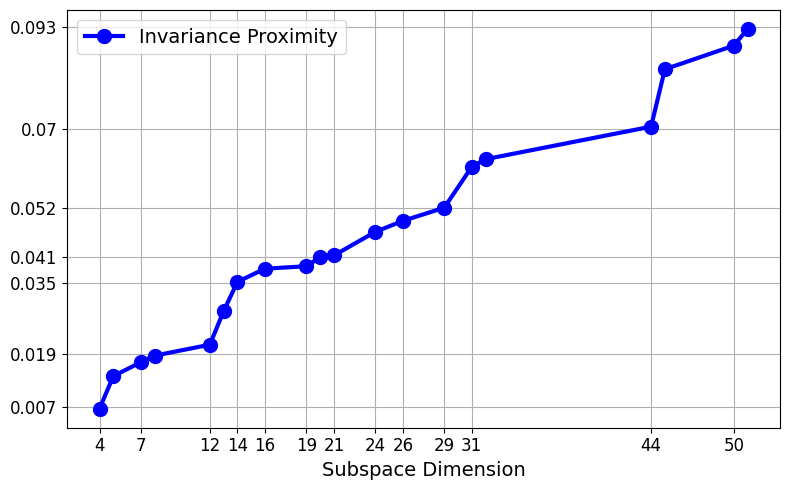}
  \caption{Invariance analysis for the Van Der Pol oscillator. The
    \texttt{MPV-SPV} pruning algorithm generates a sequence of nested
    subspaces, visualizing the trade-off between model complexity
    (dimension) and invariance proximity (IPT). Only subspaces with
    IPT below the tolerance $\epsilon_{\text{tol}} = 0.1$ are shown.}
  \label{fig:vdp_pruning}
\end{figure}
%
%

Next, we compare the prediction accuracy of the Koopman models
constructed using the pruned subspace, kernel EDMD, and standard
EDMD. We pick the pruned subspace with dimension $s = 12$ for
constructing the Koopman model. The time domain simulation results are
summarized in Figure \ref{fig:vdp_combined_results}. We pick
$x_{\text{init}} = [2.97, -3.76]^{\top}$ as the initial condition and simulate the
system for 3000 time steps. The lifted state is initialized as
$z_0 = \Psi(x_0)$, where $\Psi$ is the basis corresponding to each
method. Then, the lifted trajectory is computed using the lifted
dynamics \eqref{eq:lifted_dynamics2} and the state trajectory is
reconstructed using \eqref{eq:lifted_dynamics2}.

Let $x_{\text{true}}(t)$ denote the true state trajectory,
$x_{\text{pred}}(t)$ denote the predicted state trajectory using the
lifted model and $z_{\text{pred}}(t)$ denote the lifted state
trajectory predicted by the Koopman model. We compute the state
prediction error as
$e_{\text{state}}(t) = \| x_{\text{true}}(t) - x_{\text{pred}}(t) \|$
and the lifted state prediction error as
$e_{\text{lifted}}(t) = \| z_{\text{true}}(t) - z_{\text{pred}}(t) \|$,
where $z_{\text{true}}(t) = \Psi(x_{\text{true}}(t))$.

We observe that the pruned model is able to capture the long-term
behavior of the system more accurately than both kernel EDMD and
standard EDMD. For a shorter horizon, kernel EDMD provides better
state estimates, but the model is unstable and the error grows rapidly
over time. Standard EDMD provides consistent but less accurate
predictions compared to the pruned model. The lifted state error plot
highlights the advantage of the pruned model in maintaining a more
accurate representation of the system's dynamics in the lifted space,
which translates to better long-term predictions in the original state
space.

\begin{figure}[htb!]
  \centering
  \begin{subfigure}[b]{0.49\linewidth}
    \centering
    \includegraphics[width=\linewidth]{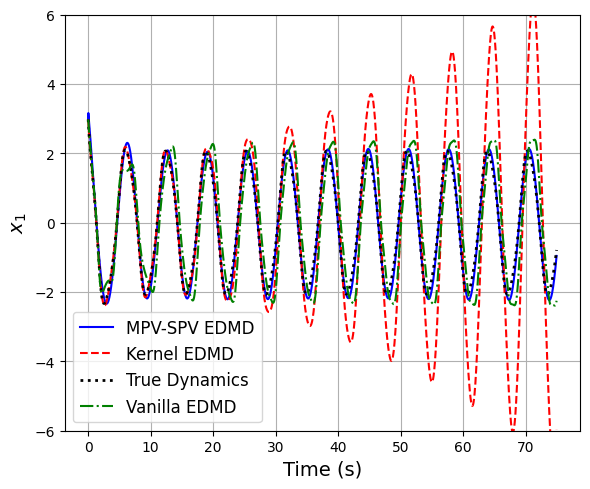} 
    \caption{$x_1$ trajectory}
    \label{fig:traj_vdp_x}
  \end{subfigure}
  \hfill
  \begin{subfigure}[b]{0.49\linewidth}
    \centering
    \includegraphics[width=\linewidth]{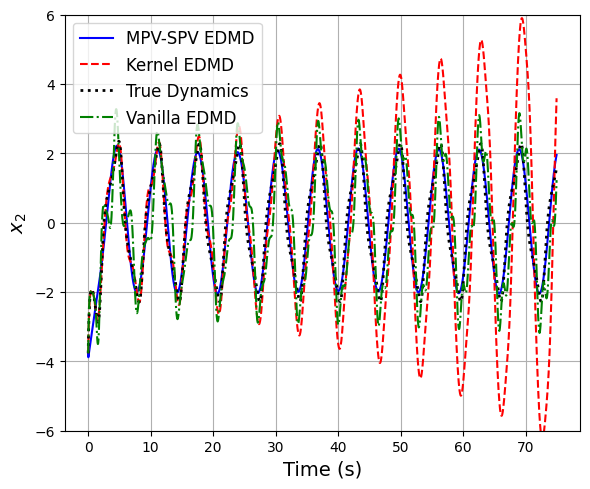} 
    \caption{$x_2$ trajectory}
    \label{fig:traj_vdp_y}
  \end{subfigure}
  \\[1ex]
  \begin{subfigure}[b]{0.49\linewidth}
    \centering
    \includegraphics[width=\linewidth]{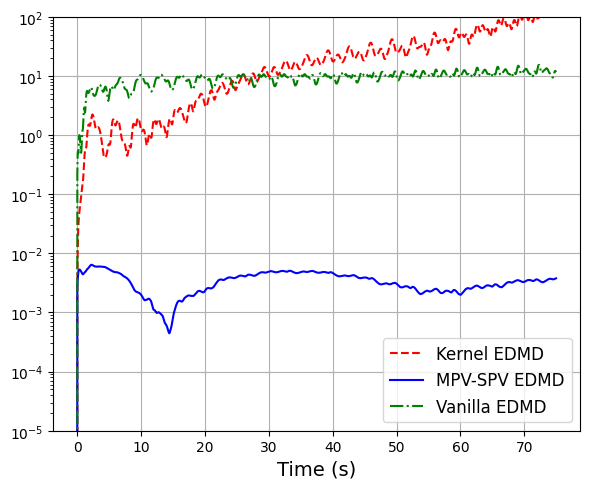} 
    \caption{Lifted state error}
    \label{fig:error_lifted}
  \end{subfigure}
  \hfill
  \begin{subfigure}[b]{0.49\linewidth}
    \centering
    \includegraphics[width=\linewidth]{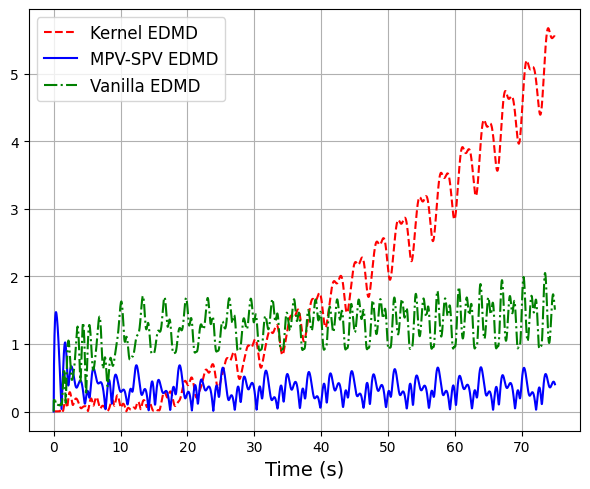} 
    \caption{State prediction error}
    \label{fig:error_state}
  \end{subfigure}
  \caption{Performance comparison of Koopman models on the Van Der Pol
    oscillator using the pruned subspace ($s=12$), Kernel EDMD
    ($s=289$), and Standard EDMD ($s=289$). The top row shows
    trajectory reconstruction, where the pruned model captures
    long-term behavior. The bottom row highlights prediction errors;
    the pruned model maintains lower error over time compared to the
    baselines, particularly in the lifted state representation.}
  \label{fig:vdp_combined_results}
\end{figure}
%
%
\vspace{-0.5cm}
\section{Conclusions}\label{sec:conclusions}
We have established a unified geometric framework for identifying
approximate Koopman-invariant subspaces via single (SPV) and multiple
(MPV) principal vector pruning.  Our formal characterization of how
approximate and external eigenfunctions are retained by these subspace
pruning methods has led us to strategically combine their strengths to
actively mitigate numerical drift. The resulting hybrid approach
ensures the accurate identification of invariant subspaces, even when
starting from massively overparameterized initial dictionaries. To
guarantee scalability, we have incorporated an efficient rank-one
update scheme; this advancement reduces the computational complexity
of tracking principal angles by an order of magnitude while adding
minimal overhead compared to the initial SVD computation.
Furthermore, we have outlined a procedure for constructing lifted
linear models that balances invariance and state reconstruction
errors. Numerical evaluation shows superior stability and accuracy in
capturing long-term dynamics compared to standard and kernel EDMD.
Future work will extend the algebraic framework to systems with
control inputs, integrate the proposed techniques with deep learning
to iteratively refine pruned dictionaries, and derive pointwise error
bounds for state-dependent certification in safety-critical
applications.

\section*{Acknowledgments}
The authors would like to thank Dr. Masih Haseli for
insightful discussions on RFB-EDMD and principal angles.


%
%

%
%

\vspace{-0.5cm}
\bibliographystyle{ieeetr}
\bibliography{../bib/alias,../bib/JC,../bib/Main,../bib/Main-add}

\begin{thebibliography}{10}

\bibitem{DS-JC:26-cdc1}
D.~Shah and J.~Cort\'es, ``A unified algebraic framework for subspace pruning in {K}oopman operator approximation via principal vectors,'' in {\em {IEEE} Conf.\ on Decision and Control}, (Honolulu, Hawaii), Dec. 2026.
\newblock Submitted.

\bibitem{BOK-JVN:32}
B.~O. {K}oopman and J.~V. Neumann, ``Dynamical systems of continuous spectra,'' {\em Proceedings of the National Academy of Sciences}, vol.~18, no.~3, pp.~255--263, 1932.

\bibitem{SK-FN-SP-JHN-CC-CS:20}
S.~Klus, F.~N{\"u}ske, S.~Peitz, J.~H. Niemann, C.~Clementi, and C.~Sch{\"u}tte, ``Data-driven approximation of the {K}oopman generator: Model reduction, system identification, and control,'' {\em Physica D: Nonlinear Phenomena}, vol.~406, p.~132416, 2020.

\bibitem{AD-DL-DS-AT:21}
A.~Dotto, D.~Lengani, D.~Simoni, and A.~Tacchella, ``Dynamic mode decomposition and {K}oopman spectral analysis of boundary layer separation-induced transition,'' {\em Physics of Fluids}, vol.~33, no.~10, 2021.

\bibitem{SPN-SG-SK-SP-KA-YW-SC:22}
S.~P. Nandanoori, S.~Guan, S.~Kundu, S.~Pal, K.~Agarwal, Y.~Wu, and S.~Choudhury, ``Graph neural network and {K}oopman models for learning networked dynamics: A comparative study on power grid transients prediction,'' {\em IEEE Access}, vol.~10, pp.~32337--32349, 2022.

\bibitem{LS-MH-GM-DB-IA-TM-JC-KK:26}
L.~Shi, M.~Haseli, G.~Mamakoukas, D.~Bruder, I.~Abraham, T.~Murphey, J.~Cort\'es, and K.~Karydis, ``Koopman operators in robot learning,'' {\em IEEE Transactions on Robotics}, vol.~42, pp.~1088--1107, 2026.

\bibitem{AM-IM:16}
A.~Mauroy and I.~Mezi{\'c}, ``Global stability analysis using the eigenfunctions of the {K}oopman operator,'' {\em IEEE Transactions on Automatic Control}, vol.~61, no.~11, pp.~3356--3369, 2016.

\bibitem{BY-IRM:24}
B.~Yi and I.~R. Manchester, ``On the equivalence of contraction and {K}oopman approaches for nonlinear stability and control,'' {\em IEEE Transactions on Automatic Control}, vol.~69, no.~7, pp.~4336--4351, 2024.

\bibitem{SAD-DVD:23}
S.~A. Deka and D.~V. Dimarogonas, ``Supervised learning of {L}yapunov functions using {L}aplace averages of approximate {K}oopman eigenfunctions,'' {\em IEEE Control Systems Letters}, vol.~7, pp.~3072--3077, 2023.

\bibitem{CMZ-AM:25}
C.~M. Zagabe and A.~Mauroy, ``Uniform global stability of switched nonlinear systems in the {K}oopman operator framework,'' {\em SIAM Journal on Control and Optimization}, vol.~63, no.~1, pp.~472--501, 2025.

\bibitem{YM-RZ-JL:23}
Y.~Meng, R.~Zhou, and J.~Liu, ``Learning regions of attraction in unknown dynamical systems via {Z}ubov-{K}oopman lifting: Regularities and convergence,'' {\em IEEE Transactions on Automatic Control}, 2025.

\bibitem{MK-IM:20}
M.~Korda and I.~Mezic, ``Optimal construction of {K}oopman eigenfunctions for prediction and control,'' {\em IEEE Transactions on Automatic Control}, vol.~65, no.~12, pp.~5114--5129, 2020.

\bibitem{MH-JC:26-auto}
M.~Haseli and J.~Cort\'es, ``Modeling nonlinear control systems via {K}oopman control family: universal forms and subspace invariance proximity,'' {\em Automatica}, vol.~185, p.~112722, 2026.

\bibitem{MH-IM-JC:25-tac}
M.~Haseli, I.~Mezi\'c, and J.~Cort\'es, ``Two roads to {K}oopman operator theory for control: infinite input sequences and operator families,'' {\em IEEE Transactions on Automatic Control}, 2025.
\newblock Submitted.

\bibitem{DG-VK-FP:24}
D.~Gadginmath, V.~Krishnan, and F.~Pasqualetti, ``Data-driven feedback linearization using the {K}oopman generator,'' {\em IEEE Transactions on Automatic Control}, vol.~69, no.~12, pp.~8844--8851, 2024.

\bibitem{VZ-EB:23}
V.~Zinage and E.~Bakolas, ``Neural {K}oopman control barrier functions for safety-critical control of unknown nonlinear systems,'' in {\em {A}merican {C}ontrol {C}onference}, pp.~3442--3447, IEEE, 2023.

\bibitem{MEV-CNJ-BH:21}
M.~E. Villanueva, C.~N. Jones, and B.~Houska, ``Towards global optimal control via {K}oopman lifts,'' {\em Automatica}, vol.~132, p.~109610, 2021.

\bibitem{JH-KC:24}
J.~Hespanha and K.~{\c{C}}amsari, ``Markov chain {M}onte {C}arlo for {K}oopman-based optimal control,'' {\em IEEE Control Systems Letters}, vol.~8, pp.~1901--1906, 2024.

\bibitem{RS-MS-KW-JB-FA:26}
R.~Strässer, M.~Schaller, K.~Worthmann, J.~Berberich, and F.~Allgöwer, ``Safedmd: A {K}oopman-based data-driven controller design framework for nonlinear dynamical systems,'' {\em Automatica}, vol.~185, p.~112732, 2026.

\bibitem{RS-JB-FA:25}
R.~Strässer, J.~Berberich, and F.~Allgöwer, ``Koopman-based control using sum-of-squares optimization: Improved stability guarantees and data efficiency,'' {\em European Journal of Control}, vol.~86, p.~101367, 2025.
\newblock Special Issue on the European Control Conference 2025.

\bibitem{DS-JC:25-csl}
D.~Shah and J.~Cort\'es, ``Controller design for bilinear neural feedback loops,'' {\em IEEE Control Systems Letters}, vol.~9, pp.~1712--1717, 2025.

\bibitem{RS-KW-IM-JB-MS-FA:26}
R.~Str{\"a}sser, K.~Worthmann, I.~Mezi{\'c}, J.~Berberich, M.~Schaller, and F.~Allg{\"o}wer, ``An overview of {K}oopman-based control: From error bounds to closed-loop guarantees,'' {\em Annual Reviews in Control}, vol.~61, p.~101035, 2026.

\bibitem{ZL-NO-EDS:25}
Z.~Liu, N.~Ozay, and E.~D. Sontag, ``Properties of immersions for systems with multiple limit sets with implications to learning {K}oopman embeddings,'' {\em Automatica}, vol.~176, p.~112226, 2025.

\bibitem{PJS:10}
P.~J. Schmid, ``Dynamic mode decomposition of numerical and experimental data,'' {\em Journal of Fluid Mechanics}, vol.~656, pp.~5--28, 2010.

\bibitem{MOW-IGK-CWR:15}
M.~O. Williams, I.~G. Kevrekidis, and C.~W. Rowley, ``A data-driven approximation of the {K}oopman operator: Extending dynamic mode decomposition,'' {\em Journal of Nonlinear Science}, vol.~25, no.~6, pp.~1307--1346, 2015.

\bibitem{MK-IM:18}
M.~Korda and I.~Mezi{\'c}, ``On convergence of extended dynamic mode decomposition to the {K}oopman operator,'' {\em Journal of Nonlinear Science}, vol.~28, no.~2, pp.~687--710, 2018.

\bibitem{FN-SP-FP-MS-KW:23}
F.~N{\"u}ske, S.~Peitz, F.~Philipp, M.~Schaller, and K.~Worthmann, ``Finite-data error bounds for {K}oopman-based prediction and control,'' {\em Journal of Nonlinear Science}, vol.~33, no.~1, p.~14, 2023.

\bibitem{FK-FMP-MS-AS-KW:25}
F.~K\"{o}hne, F.~M. Philipp, M.~Schaller, A.~Schiela, and K.~Worthmann, ``\(\boldsymbol{L}^{\boldsymbol{\infty }}\)-error bounds for approximations of the {K}oopman operator by kernel extended dynamic mode decomposition,'' {\em SIAM Journal on Applied Dynamical Systems}, vol.~24, no.~1, pp.~501--529, 2025.

\bibitem{MH-FMP-MS-KW:25}
M.~Hertel, F.~M. Philipp, M.~Schaller, and K.~Worthmann, ``Koopman for stochastic dynamics: error bounds for kernel extended dynamic mode decomposition,'' {\em arXiv preprint arXiv:2512.20247}, 2025.

\bibitem{MH-JC:23-auto}
M.~Haseli and J.~Cort\'es, ``Generalizing dynamic mode decomposition: balancing accuracy and expressiveness in {K}oopman approximations,'' {\em Automatica}, vol.~153, p.~111001, 2023.

\bibitem{MJC-CD-AH:25}
M.~J. Colbrook, C.~Drysdale, and A.~Horning, ``Rigged dynamic mode decomposition: Data-driven generalized eigenfunction decompositions for {K}oopman operators,'' {\em SIAM Journal on Applied Dynamical Systems}, vol.~24, no.~2, pp.~1150--1190, 2025.

\bibitem{YX-KS-NKL-ZS:25}
Y.~Xu, K.~Shao, N.~K. Logothetis, and Z.~Shen, ``Reskoopnet: Learning {K}oopman representations for complex dynamics with spectral residuals,'' in {\em Forty-second International Conference on Machine Learning}, (Vancouver, Canada), 2025.

\bibitem{BG-JP-SC-OA:25}
B.~Gao, J.~Patracone, S.~Chr{\'e}tien, and O.~Alata, ``Conformal online learning of deep {K}oopman linear embeddings,'' {\em arXiv preprint arXiv:2511.12760}, 2025.

\bibitem{HZ-STMD-CWR-EAD-LNC:20}
H.~Zhang, S.~T.~M. Dawson, C.~W. Rowley, E.~A. Deem, and L.~N. Cattafesta, ``Evaluating the accuracy of the dynamic mode decomposition,'' {\em Journal of Computational Dynamics}, vol.~7, no.~1, pp.~35--56, 2020.

\bibitem{SI-SB-AA-NK:25}
S.~M. Ichinaga, S.~L. Brunton, A.~Y. Aravkin, and J.~N. Kutz, ``Sparse-mode dynamic mode decomposition for disambiguating local and global structures,'' {\em arXiv preprint arXiv:2507.19787}, 2025.

\bibitem{MJC-LJA-MS:23}
M.~J. Colbrook, L.~J. Ayton, and M.~Sz{\H{o}}ke, ``Residual dynamic mode decomposition: robust and verified {K}oopmanism,'' {\em Journal of Fluid Mechanics}, vol.~955, p.~A21, 2023.

\bibitem{MJC-AT:24}
M.~J. Colbrook and A.~Townsend, ``Rigorous data-driven computation of spectral properties of {K}oopman operators for dynamical systems,'' {\em Communications on Pure and Applied Mathematics}, vol.~77, no.~1, pp.~221--283, 2024.

\bibitem{MJK:24}
M.~J. Colbrook, ``Another look at residual dynamic mode decomposition in the regime of fewer snapshots than dictionary size,'' {\em Physica D: Nonlinear Phenomena}, vol.~469, p.~134341, 2024.

\bibitem{GC-NB-JCL-SB-MJC:26}
G.~Conradie, N.~Boull{\'e}, J.-C. Loiseau, S.~L. Brunton, and M.~J. Colbrook, ``Trustworthy {K}oopman {O}perator {L}earning: {I}nvariance {D}iagnostics and {E}rror {B}ounds,'' {\em arXiv preprint arXiv:2603.15091}, 2026.

\bibitem{MH-JC:23-csl}
M.~Haseli and J.~Cort\'es, ``Temporal forward-backward consistency, not residual error, measures the prediction accuracy of {E}xtended {D}ynamic {M}ode {D}ecomposition,'' {\em IEEE Control Systems Letters}, vol.~7, pp.~649--654, 2023.

\bibitem{MH-JC:24-csl-arxiv-revised}
M.~Haseli and J.~Cort\'es, ``Invariance proximity: closed-form error bounds for finite-dimensional {K}oopman-based models,'' {\em https://arxiv.org/abs/2311.13033}, 2024.

\bibitem{MH-JC:22-tac}
M.~Haseli and J.~Cort\'es, ``Learning {K}oopman eigenfunctions and invariant subspaces from data: {S}ymmetric {S}ubspace {D}ecomposition,'' {\em IEEE Transactions on Automatic Control}, vol.~67, no.~7, pp.~3442--3457, 2022.

\bibitem{MH-JC:25-access}
M.~Haseli and J.~Cort\'es, ``Recursive forward-backward {EDMD}: Guaranteed algebraic search for {K}oopman invariant subspaces,'' {\em IEEE Access}, vol.~13, pp.~61006--61025, 2025.

\bibitem{AM-YS-IM:20}
A.~Mauroy, Y.~Susuki, and I.~Mezi{\'c}, {\em Koopman Operator in Systems and Control}.
\newblock New York: Springer, 2020.

\bibitem{lQL-FD-EMB-IGK:17}
Q.~Li, F.~Dietrich, E.~M. Bollt, and I.~G. Kevrekidis, ``Extended dynamic mode decomposition with dictionary learning: A data-driven adaptive spectral decomposition of the {K}oopman operator,'' {\em Chaos}, vol.~27, no.~10, p.~103111, 2017.

\bibitem{AB-GHG:73}
A.~Bj{\"o}rck and G.~H. Golub, ``Numerical methods for computing angles between linear subspaces,'' {\em Mathematics of Computation}, vol.~27, no.~123, pp.~579--594, 1973.

\bibitem{GHG-CFVL:13}
G.~H. Golub and C.~F.~V. Loan, {\em Matrix Computations}.
\newblock The Johns Hopkins University Press, 2013.

\bibitem{EA-LAPACK:99}
E.~Anderson, Z.~Bai, C.~Bischof, L.~S. Blackford, J.~Demmel, J.~Dongarra, J.~Du~Croz, A.~Greenbaum, S.~Hammarling, A.~McKenney, {\em et~al.}, {\em LAPACK users' guide}.
\newblock SIAM, 1999.

\end{thebibliography}

\section{Appendix}

\subsection{Results on principal angles and vectors}
We present some additional results on principal angles and vectors
that are used in the main text.


\begin{lemma}[Alternate Characterization of Principal
  Arguments]\label{lemma:cos_theta_max}
  Let $\Uc, \Vc \subset \Hc$ be two subspaces with
  $a = \dim(\Uc) \leq \dim(\Vc) = b$.  Let $\{\theta_j\}_{j=1}^a$ be
  the principal angles between $\Uc$ and $\Vc$, and let
  $\{x_j^{\Uc}\}_{j=1}^a \subset \Uc$ and
  $\{y_j^{\Vc}\}_{j=1}^a \subset \Vc$ be the corresponding principal
  vectors.  Then, for $k \in [a]$,
  \begin{equation}\label{eq:alternate_deftn_pvs} 
    \cos \theta_{a - (k-1)} = \min_{\substack{x \in \Uc \\ x \perp
        x_{a-(k-2)}^{\Uc},\dots,x_a^{\Uc}}} \frac{\| \Pc_{\Vc}(x)
      \|}{\| x \|}  ,
  \end{equation}
  and for a given $k$, the minimizer of \eqref{eq:alternate_deftn_pvs}
  is the principal vector $x_{a-(k-1)}^{\Uc}$. Consequently, the set
  of all minimizers as $k \in [a]$ is exactly the set of principal
  vectors $\{x_j^{\Uc}\}_{j=1}^{a}$.
\end{lemma}

\begin{proof}
  Any $x \in \Uc$ can be written as
  $x = \sum_{i = 1}^{a} c_i x_i^{\Uc}$. Using the definition of
  principal angles and vectors, this implies that
  $\Pc_{\Vc} x = \sum_{i = 1}^{a} c_i y_i^{\Vc} \cos \theta_i$. With
  $\Lambda = \operatorname{diag}(\cos \theta_i)_{i=1}^a$, we have
  \begin{equation*}
    \frac{\| \Pc_{\Vc}(x) \|^2}{\| x \|^2} = \frac{\sum_{i = 1}^{a}
      c_i^2 \cos^2 \theta_i}{\sum_{i = 1}^{a} c_i^2} = \frac{c^{\top}
      \Lambda^2 c}{c^{\top} c}  
  \end{equation*}
  For $k=1$, there are no orthogonality constraints. Clearly, the
  minimum value of the expression above is $\cos^2 \theta_a$, which is
  achieved when $c = e_a$ (corresponding to $x = x_a^{\Uc}$).
  
  For a general $k > 1$, we impose the orthogonality constraints
  $x \perp x_{a-(k-2)}^{\Uc}, \dots, x_a^{\Uc}$, which implies
  $c \perp e_{a-(k-2)}, \dots, e_a$. Thus, under these orthogonality
  constraints, the minimum value of the expression is
  $\cos^2 \theta_{a-(k-1)}$, which is achieved when $c = e_{a-(k-1)}$
  (corresponding to $x = x_{a-(k-1)}^{\Uc}$).
\end{proof}

Note that Lemma~\ref{lemma:cos_theta_max} expresses $\cos\theta_i$ via
projections of vectors in the lower-dimensional subspace $\Uc$ onto
the higher-dimensional subspace $\Vc$. The reverse
operation—projecting vectors in $\Vc$ onto $\Uc$ and minimizing the
analogous ratio—does not, in general, yield $\cos\theta_i$.

The following result bounds the norm of the projection onto $\Vc$ of a
vector in $\Uc$ using the principal angles.

\begin{corollary}[Bounds on Projection
  Norms]\label{corollary:bounds_proj_norm}
  Consider two subspaces $\Uc, \Vc \subset \Hc$ with
  $\text{dim}(\Uc) \leq \text{dim}(\Vc)$ as defined in Lemma
  \ref{lemma:cos_theta_max}. Then, for $u \in \Uc$,
  \begin{equation*}
    \cos^2 \theta_{\max}(\Uc, \Vc) \leq \frac{\| \Pc_{\Vc}(u) \|^2}{\|
      u \|^2} \leq \cos^2 \theta_{\min}(\Uc, \Vc) .
  \end{equation*}
\end{corollary}

\subsection{Perturbation bounds}
We describe a result bounding the sensitivity of the angle between two
vectors under perturbation. This is used in the proof of Theorem
\ref{thm:stability_MPV_pruning}.

\begin{lemma}[Perturbation Bound]\label{lemma:Perturbation_bound}
  Let $x, x', y, y'$ be vectors such that $\|x\| \ge \alpha_x$ and
  $\|y\| \ge \alpha_y$. Let $\theta(x,y)$ denote the angle between $x$
  and $y$.
  %
  %
  The difference in the sine of the angles is bounded by the
  vector perturbations:
  \begin{align*}
    |\sin \theta(x,y) - \sin \theta(x',y')| \le \frac{2}{\alpha_x}\|x -
    x'\| + \frac{4}{\alpha_y}\|y - y'\|.     
  \end{align*}
\end{lemma}
\begin{proof}
  Let $u, u', v, v'$ be the normalized vectors (e.g., $u =
  x/\|x\|$). By definition, $\sin \theta(x,y) = \|(I - P_v)u\|$. We
  define the difference operator
  $\Delta = (I - P_{v'})u' - (I - P_v)u$. Using the triangle
  inequality:
  \begin{align*}
    &|\sin \theta(u', v') - \sin \theta(u, v)| \le \|\Delta\| \,\,
      \le \|(I - P_{v'})(u' - u)\| \\
    &+ \|(P_v - P_{v'})u\| \,\,
      \le \|u' - u\| + 2\|v - v'\| \tag{1}
  \end{align*}
  where we used the bound $\|(P_v - P_{v'})u\| \le 2\|v - v'\|$
  derived from symmetric projection properties.
  
  To relate this to unnormalized vectors, we apply the Lipschitz
  continuity of the map $x \mapsto x/\|x\|$. For $\|x\| \ge \alpha_x$:
  \begin{align*}
    \|u - u'\| = \left\| \frac{x}{\|x\|} - \frac{x'}{\|x'\|} \right\|
    \le \frac{2}{\alpha_x} \|x - x'\| \tag{2} 
  \end{align*}
  Substituting (2) into (1) yields the final result.
\end{proof}

\vspace*{-4ex}

\begin{IEEEbiography}[{\includegraphics[width=1in,height=1.2in,clip,keepaspectratio,angle
    = 90]{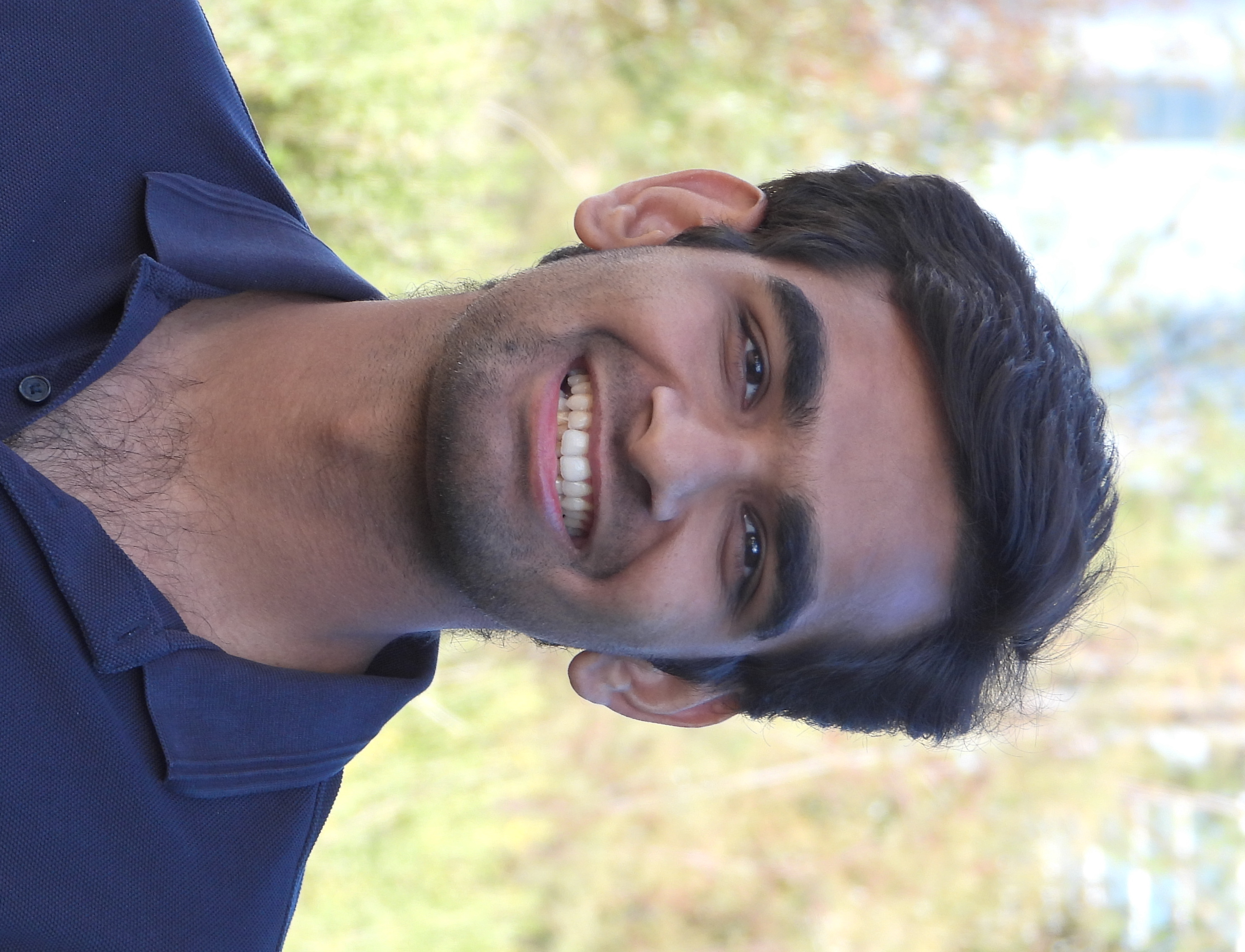}}]{Dhruv Shah} is a Ph.D. student at the
  University of California, San Diego, where he also completed his
  Master's in Mechanical and Aerospace Engineering (2024). Previously,
  he graduated from the Indian Institute of Technology, Bombay with a
  Bachelor's in Electrical Engineering (2023). His work explores
  nonlinear systems and Koopman operator theory, specifically focusing
  on their utility in robotics, aerospace engineering, and machine
  learning.
\end{IEEEbiography}

\vspace*{-4ex}

\begin{IEEEbiography}[{\includegraphics[width=1in,height=1.2in,clip,keepaspectratio]{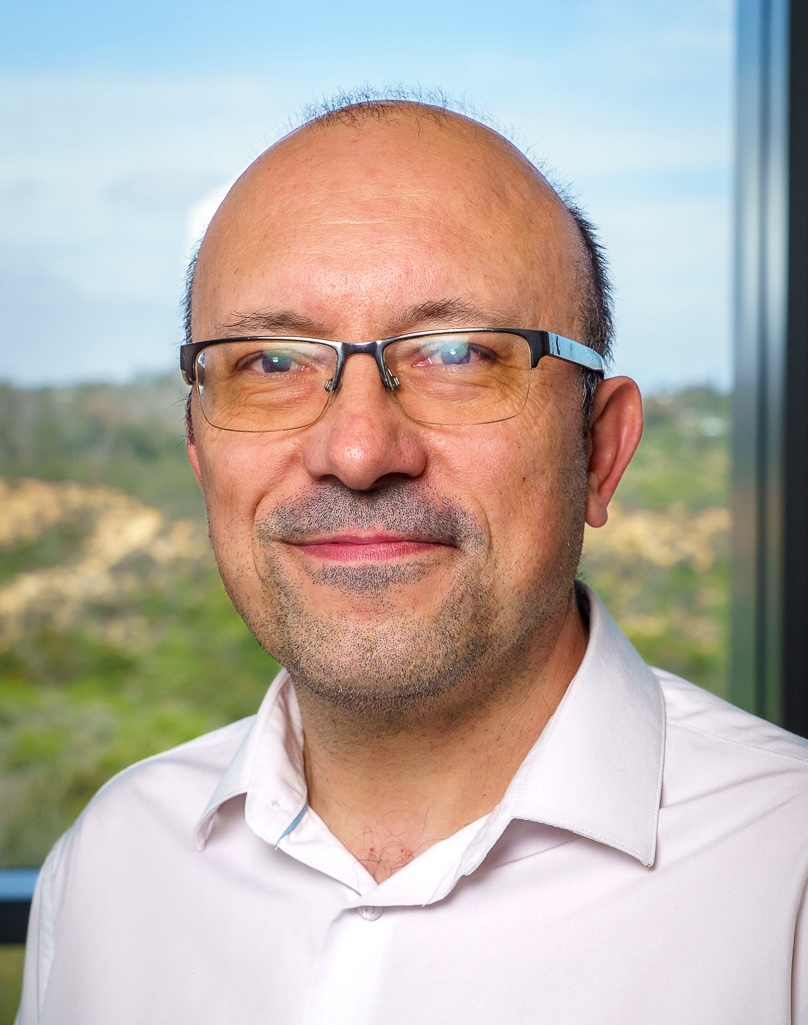}}]{Jorge
    Cort\'{e}s}(M'02, SM'06, F'14) received the Licenciatura degree in
  mathematics from Universidad de Zaragoza, Spain, in 1997, and the
  Ph.D. degree in engineering mathematics from Universidad Carlos III
  de Madrid, Spain, in 2001. He held postdoctoral positions with the
  University of Twente, Twente, The Netherlands, and the University of
  Illinois at Urbana-Champaign, Illinois, USA.  He is a Professor and
  Cymer Corporation Endowed Chair in High Performance Dynamic Systems
  Modeling and Control at the Department of Mechanical and Aerospace
  Engineering, UC San Diego, California, USA.  He is a Fellow of IEEE,
  SIAM, and IFAC.  His research interests include distributed control
  and optimization, network science, autonomy, learning, nonsmooth
  analysis, decision making under uncertainty, network neuroscience,
  and multi-agent coordination in robotic, power, and transportation
  networks.
\end{IEEEbiography}

\end{document}